\documentclass[ejs]{imsart}
\RequirePackage{natbib}
\arxiv{math.PR/0000000}

\usepackage[english]{babel}
\usepackage[iso-8859-7]{inputenc}
\usepackage[T1]{fontenc}
\usepackage{amsthm}
\usepackage{graphicx}
\usepackage{slashbox}
\usepackage[center]{caption}
\usepackage{amsmath}
\usepackage{amssymb}
\usepackage{dsfont}
\usepackage{here}
\usepackage{rotating}
\usepackage{arydshln}
\usepackage{geometry}
\usepackage{blkarray}
\usepackage{algorithm}
\usepackage{algorithmic}
\usepackage{url}

\DeclareMathOperator*{\argmin}{arg\,min}

\newtheorem{theorem}{Theorem}
\newtheorem{lemma}{Lemma}
\newtheorem{proposition}{Proposition}
\newtheorem{example}{Example}

\newcommand{\overbar}[1]{\mkern 1.5mu\overline{\mkern-1.5mu#1\mkern-1.5mu}\mkern 1.5mu}

\begin{document}

\begin{frontmatter}

\title{Delete or Merge Regressors for Linear Model Selection} 
\runtitle{DMR for Linear Model Selection}

\begin{aug}
  \author{Aleksandra Maj-Ka\'nska \thanksref{t2} \ead[label=e1]{a.maj@phd.ipipan.waw.pl}}
 \address{
	Institute of Computer Science \\
	Polish Academy of Sciences \\
	Jana Kazimierza 5 \\
	01-248 Warsaw \\
	Poland \\
        \printead{e1}}
  \author{Piotr Pokarowski \thanksref{t3} \ead[label=e2]{pokar@mimuw.edu.pl}}
 \address{
	Faculty of Mathematics, Informatics and Mechanics \\
	University of Warsaw \\
	Banacha 2 \\
	02-097 Warsaw \\
	Poland \\
	\printead{e2}}
  \and
    \author{Agnieszka Prochenka \thanksref{t2}
  \ead[label=e3]{a.prochenka@phd.ipipan.waw.pl}}
 \address{
	Institute of Computer Science \\
	Polish Academy of Sciences \\
	Jana Kazimierza 5 \\
	01-248 Warsaw \\
	Poland \\
        \printead{e3}}

  \thankstext{t2}{Study was supported by research fellowship within "Information technologies: research and their interdisciplinary applications" agreement number POKL.04.01.01-00-051/10-00.}
 \thankstext{t3}{Study was supported by Polish National Science Center grant 2011/01/B/NZ2/00864}

  \runauthor{A. Maj-Ka\'nska et al.}

\end{aug}

\begin{abstract}
We consider a problem of linear model selection in the presence of both continuous and categorical predictors. Feasible models consist of subsets of numerical variables and partitions of levels of factors. A new algorithm called delete or merge regressors (DMR) is presented which is a stepwise backward procedure involving ranking the predictors according to squared t-statistics and choosing the final model minimizing BIC. In the article we prove consistency of DMR when the number of predictors tends to infinity with the sample size and describe a simulation study using a pertaining  \texttt{R} package. The results indicate significant advantage in time complexity and selection accuracy of our algorithm over Lasso-based methods described in the literature. Moreover, a  version of DMR for generalized linear models is proposed.
\end{abstract}



\tableofcontents

\end{frontmatter}

\section{Introduction}

Model selection is usually understood as selection of continuous explanatory variables.  
However, when a categorical predictor is considered, in order to reduce model's complexity, 
we can either exclude the whole factor or merge its levels.

A traditional method of examining the relationship between a continuous response and categorical variables is analysis of variance (ANOVA). 
After detecting the overall importance of a factor, pairwise comparisons of group means are used to test significance of differences between its levels. 
Typically post-hoc analysis such as Tukey's honestly significant difference (HSD) test or multiple comparison adjustments (Bonferroni, Scheffe) are used. A drawback of pairwise comparisons is non-transitivity of conclusions. 

For example, let us consider data \texttt{barley} from \texttt{R} library \texttt{lattice} discussed already in \cite{bondell}. Total yield of barley for 5 varieties at 6 sites in each of two years  is modeled.  The dependence between the response and the varieties variable with the use of Tukey's HSD analysis (Figure~\ref{rystukey}) gives inconclusive answers:  $\beta_P = \beta_M$, $\beta_P = \beta_T$, but $\beta_T \neq \beta_M$. 

In this work we introduce a novel procedure called delete or merge regressors (DMR), which enables efficient search among partitions of factor levels, for which the issue of non-transitivity does not occur. If we apply DMR to the \texttt{barley} data, we get the following partition of varieties: $\{\{S,M,V,P \}, \{T \} \}$. Detailed description of the data set and the characteristics of the chosen model can be found in Section \ref{subsecCARS}.

\begin{figure}[!ht]
\centerline{%
\includegraphics[width=80mm]{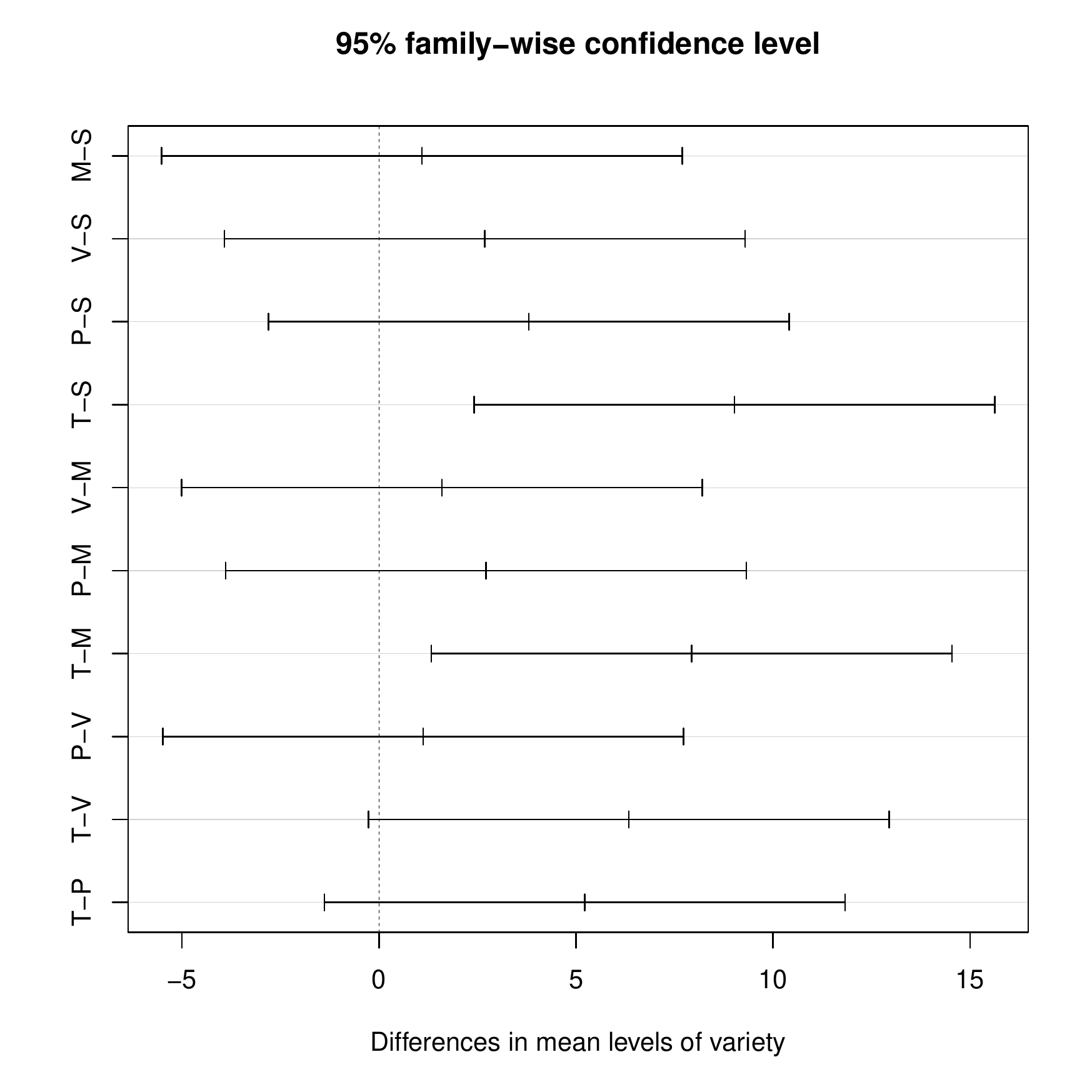}}
\caption{Results of Tukey's HSD.}\label{rystukey}
\end{figure}

The idea of partitioning a set of levels of a factor into non-overlapping groups has already been discussed in the literature. 
In the article \cite{tukey} a stepwise backward procedure based on the studentized range which gives grouping of means for samples from normal distributions was proposed. 
Other methods of clustering of sample means were described in \cite{scottknott}, where the set of means is partitioned from coarsest to finest, and in \cite{calinski} whose algorithm adapts hierarchical clustering to the problem. 
In more recent articles \cite{porreca} and \cite{ciampi} efficient algorithms for datasets partitioning using generalized likelihood ratio test can be found. 
However, all the mentioned methods assume an arbitrary choice of significance level for the underlying tests. 
In our procedure we avoid the problem by selecting the final partition according to minimal value of information criterion. 

Information criterion as an objective function for partition selection is used in the procedures described in \cite{dayton03}. 
Dayton's SAS procedure,
called paired comparisons information criteria (PCIC), computes AIC and BIC values for all ordered subsets of independent means for both homogeneous and heterogeneous models.
In contrast to DMR these methods do not allow for simultaneous factor partitioning and selection of continuous variables. 

A method introduced in \cite{bondell} called collapsing and shrinkage ANOVA (CAS-ANOVA) solves the same problem as DMR with use of the least absolute shrinkage and selection operator (Lasso; \cite{lasso}), where the $L_1$ penalty is imposed on differences between parameters corresponding to levels of each factor. 
This algorithm can be interpreted as a generalization of fused Lasso (\cite{fusedlasso}) to data with categorical variables. 
In \cite{tutz} one can find a modification of CAS-ANOVA, which is more computationally efficient because of using the least angle regression algorithm (LARS; \cite{lar}). 
Another algorithm, based on regularized model selection with categorical predictors and effect modifiers (\cite{tutzGLM}) is implemented in \texttt{R} package \texttt{gvcm.cat}. It generalizes Lasso approach to simultaneous factor partitioning and selection of continuous variables to generalized linear models. 
The algorithm is based on local quadratic approximation and iterated reweighted least squares. 

We propose a backward selection procedure called delete or merge regressors (DMR), which 
combines deleting continuous variables with merging levels of factors. 
The method employs a  greedy search among linear models with a set of constraints of two types: either a parameter for a continuous variable is set to zero or parameters corresponding to two levels of a factor are set to equal each other. In each step the choice of constraint is based on the order of squared t-statistics. 
As a result  a nested family of linear models is obtained and the final decision is made by minimization of Bayesian information criterion (BIC).
The method adapts agglomerative clustering, where squared t-statistics define the dissimilarity measure. 
This procedure generalizes concepts introduced in \cite{loh}  and \cite{ciampi} .  

In the article we show that DMR algorithm is a consistent model selection method under rather weak assumptions when $p$ tends to infinity with $n$. Furthermore, thanks to using a recursive formula for RSS in a nested family of linear models,  the time complexity of DMR algorithm is just $O(np^2)$. This makes the algorithm much faster than the competitive Lasso-based methods.
In the article we describe a simulation study and discuss a pertaining \texttt{R} package. The simulations show that DMR in comparison to adaptive Lasso methods described in the literature gives  better results in terms of accuracy without the troublesome choice of the $\lambda$ grid.

The remainder of the article proceeds as follows. The class of feasible models considered when performing model selection is defined in Section~\ref{FM}. DMR procedure is introduced in Section \ref{secALG}, while its asymptotic properties are discussed in Section \ref{secASY}. Simulations and real data examples are given in Section \ref{secSIMUL} to illustrate the method. All proofs are given in the Appendix.


\section{Feasible models}\label{FM}
In this section we first introduce some definitions regarding the form of the data and models considered. In particular, we define the set of feasible models, which are linear spaces of parameters with linear constraints and we show how by change of variables the constrained problem can be replaced by unconstrained one. Later we indicate that properties of OLS (ordinary least squares) estimators transfer to feasible models. 

\subsection{Definitions}
Let us consider data generated by a full rank linear model with $n$ observations and $p<n$ parameters:
\begin{equation}\label{model}
\bold{y} = \bold{X} \boldsymbol\beta^* + \boldsymbol\varepsilon = 
\bold{1} \beta_{00}^* + \bold{X}_0 \boldsymbol\beta_0^* + \bold{X}_1 \boldsymbol\beta_1^* + \ldots + \bold{X}_l \boldsymbol\beta_l^* + \boldsymbol\varepsilon,
\end{equation}
where: 
\begin{enumerate}
\item $\boldsymbol\varepsilon$ is a vector of iid zero-mean gaussian errors, $\boldsymbol\varepsilon \sim \mathcal{N} (\bold{0}, \sigma^2\mathds{I})$.
\item $\bold{X} = [ \bold{1},\bold{X}_0,\bold{X}_1,\ldots,\bold{X}_l]$ is a model matrix organized as follows:
$\bold{X}_0$ is a matrix corresponding to continuous regressors and
$\bold{X}_1,\ldots,\bold{X}_l$ are zero-one matrices encoding corresponding factors with the first level set as the reference.
\item $\boldsymbol\beta^* = [\beta_{00}^*,\boldsymbol\beta^{*T}_0,\boldsymbol\beta^{*T}_1, \ldots,\boldsymbol\beta^{*T}_l]^T \in \mathds{R}^p$ is a parameter vector 
organized as follows: $\beta_{00}^*$ is the intercept,
$\boldsymbol\beta_0^* = [\beta_{10}^*, \ldots, \beta_{p_00}^*]^T$ is a vector of coefficients for continuous 
variables and 
$\boldsymbol\beta_k^* = [\beta_{2k}^*, \ldots, \beta_{p_kk}^* ]^T$ is a vector of parameters corresponding 
to the $k$-th factor, $k = 1,\ldots,l$, hence the length of the parameter vector is $p =1+p_0+ (p_1-1)+ \ldots + (p_l-1)$. 
\end{enumerate} 

Denote sets of indexes: $N = \{ 0,1,\ldots, l \}$, $N_0 = \{0, 1, \ldots, p_0 \}$ and $N_k = \{ 2,3, \ldots, p_k \}$ for $k \in N \setminus \{ 0 \}$. Let us define an elementary constraint for linear model (\ref{model}) as a linear constraint of one of two types:
\begin{equation}\label{eqHi}
\mathcal{H}_{jk}: \ \beta_{jk}^*= 0 \text{ where } j \in N_k \setminus \{0\},\  k \in N,
\end{equation}
\begin{equation}\label{eqHij}
\mathcal{H}_{ijk}: \ \beta_{ik}^* = \beta_{jk}^* \text{ where }  i, j \in N_k,\ i \neq j,\  k \in N \setminus \{ 0 \}.
\end{equation}

A feasible model can be defined as a sequence $M = (P_{0},P_{1},...,P_{l})$, where $P_{0}$ denotes a subset of indexes of continuous variables and $P_{k}$ is a particular partition of levels of  the $k$-th factor. Such a model can be encoded by a set of elementary constraints. A set of all feasible models is denoted by $\mathcal{M}$. Let us denote model $F \in \mathcal{M}$ without constraints of types (\ref{eqHi}) or (\ref{eqHij}) as the full model.\\
\textbf{Example 1}. For illustration, let us consider a model with one factor and one continuous variable:
\[
\bold{y} = \bold{X} \boldsymbol\beta^* + \boldsymbol\varepsilon = \bold{1} \cdot 1 + \bold{X}_0 \cdot 2 + \bold{X}_1 \cdot  \left[\begin{array}{c}
-2 \\ 
-2  \\ 
0  \\ 
\end{array}\right]  +   \boldsymbol\varepsilon=
\]
\begin{equation}\label{illex}
=  \left[\begin{array}{c}
1\\
1\\
1\\
1\\
 1\\
 1\\
 1\\
 1\\
\end{array}\right] \cdot 1+ \left[\begin{array}{c}
-0.96\\
 -0.29\\
  0.26\\
 -1.15\\
  0.2\\
  0.03\\
 0.09\\
  1.12\\
\end{array}\right]  \cdot 2 +  \left[\begin{array}{ccc}
0 & 0 & 0  \\ 
0 & 0 & 0  \\ 
1 & 0 & 0  \\ 
1 & 0 & 0  \\ 
0 & 1 & 0  \\ 
0 & 1 & 0  \\ 
0 & 0 & 1  \\ 
0 & 0 & 1  \\ 
\end{array}\right] \left[\begin{array}{c}
-2 \\ 
-2  \\ 
0  \\ 
\end{array}\right]    +  \left[\begin{array}{c}
 -1.22 \\
1.27 \\
 -0.74 \\ 
-1.13\\
 -0.72\\
  0.25\\
  0.15\\
 -0.31\\
\end{array}\right]  ,
\end{equation}
where $ \bold{X}_0$ and $ \boldsymbol\varepsilon$ are vectors of length 8 generated independently from standard normal distribution, $ \mathcal{N} (\bold{0}, \mathds{I})$. Then $\beta^* = [1,2,-2,-2,0]^T$. The full model $F = (P_{0} = \{ 1 \},P_{1} =\{ \{ 1 \}, \{ 2 \}, \{ 3 \}, \{ 4 \} \})$ with $p_0 =1, p_1 =4, p=5$.  The model corresponding to $\beta^*$ is $(P_{0} = \{ 1 \},P_{1}= \{ \{ 1, 4 \}, \{ 2, 3 \} \})$ and is the same as $F$ with two elementary constraints: $\beta^*_{41} = 0$ and $\beta^*_{21} = \beta^*_{31} $.


\subsection{Unconstrained parametrization of feasible models}\label{subsecREGULAR}
A feasible model can be defined by a linear space of parameters
\begin{equation}\label{modcon}
\mathcal{L}_M = \left\{ \boldsymbol\beta \in \mathds{R}^p: \bold{A}_{0M} \boldsymbol\beta = 0 \right\},
\end{equation}
where $\bold{A}_{0M}$ is a $(p-q) \times p$ matrix encoding $q$ elementary constraints induced by the model.
Such a constraint matrix can be expressed in many ways. In particular, every linear space can be spanned by different  vectors. The number of such vectors can be greater than the dimension of the space when they are linearly dependent. In order to unify the form of  a constraint matrix, we introduce the notion of regular form, which is described in the Appendix~\ref{sec:regular}. We assume that $\bold{A}_{0M}$  is in regular form.
 Let $\bold{A}_{1M}$ be a $q \times p$ complement of $\bold{A}_{0M}$ to invertible matrix $A_M$, that is:
\[
\bold{A}_M = \left[\begin{array}{c}
\bold{A}_{1M} \\ \hline
\bold{A}_{0M} \\
\end{array}\right].
\]

 Denote:
\begin{equation}\label{eq:Admp}
\bold{A}_M^{-1} = \left[\begin{array}{ccc}
\bold{A}_M^1 & \vline & \bold{A}_M^0 \\
\end{array}\right],
\end{equation}
where $\bold{A}_M^1$ is a $p \times q$ matrix.
In order to replace a constrained by an unconstrained parametrization change of variables in model $M$ is performed.  Let $\boldsymbol\beta_M \in \mathcal{L}_M$ and $\boldsymbol\xi_M = \bold{A}_{1M}\boldsymbol\beta_M$. We have:
\begin{equation}\label{betaM}
\boldsymbol\beta_M = \bold{A}_M^1 \boldsymbol\xi_M.
\end{equation}
Indeed,
\[
\boldsymbol\beta_M = \bold{A}_M^{-1} \bold{A}_M \boldsymbol\beta_M =  \bold{A}_M^{-1} \left[\begin{array}{c}
\bold{A}_{1M} \boldsymbol\beta_M \\ \hline
\bold{A}_{0M} \boldsymbol\beta_M  \\  
\end{array}\right]  = \left[\begin{array}{ccc}
\bold{A}_M^1 & \vline & \bold{A}_M^0 \\
\end{array}\right]  \left[\begin{array}{c}
\boldsymbol\xi_{M} \\ \hline
\bold{0}\\  
\end{array}\right] = \bold{A}_M^1  \boldsymbol\xi_{M}.
\]
From equation (\ref{betaM}) we obtain $\bold{X} \boldsymbol\beta_M = \bold{Z}_{1M} \boldsymbol\xi_M$, where  $\bold{Z}_{1M} = \bold{X} \bold{A}^1_M$ and $\mathcal{L}_M = \{ \bold{A}_M^1 \boldsymbol\xi: \boldsymbol\xi \in~\mathds{R}^{q}\} $.
Let us notice that $\mathcal{L}_M$ is a linear space spanned by columns of $\bold{A}_{M}^1$. 
The dimension of space $\mathcal{L}_M$ will be called the size of model $M$ and denoted by $|M|$. Note that $|M| = q$. \\
\textbf{Example 1 continued}.
Matrices $\bold{A}_M, \bold{Z}_{1M}$ and $\boldsymbol\xi_M$ are:
\[
\bold{A}_M =   \left[\begin{array}{ccccc}
1 & 0 & 0 & 0 & 0  \\
0 & 1 & 0 & 0 & 0  \\
0 & 0 & 1 & 0 & 0 \\
\hline
0 & 0 & -1 & 1 & 0  \\
0 & 0 & 0 & 0 & 1  \\
\end{array}\right], \ \bold{A}_M^1 = \left[\begin{array}{ccc}
1 & 0 & 0 \\
0 & 1 & 0  \\
0 & 0 & 1 \\
0 & 0 & 1   \\
0 & 0 & 0 \\
\end{array}\right]
\]
$$
\bold{Z}_{1M} =  \left[\begin{array}{ccc}
1 & -0.96 & 0 \\ 
1& -0.29 & 0 \\ 
1 & 0.26 & 1 \\ 
1 & -1.15 & 1 \\ 
1 & 0.2 & 1   \\ 
1 & 0.03  & 1   \\ 
1 & 0.09 & 0  \\ 
1 & 1.12 & 0  \\ 
\end{array}\right] \ ,\  \boldsymbol\xi_M = (\xi_1, \xi_2, \xi_3)^T \ ,\  \xi_1 = \beta _{00}^*,\ \xi_2 = \beta_{10}^*\ ,\ \xi_3 = \beta_{21}^* = \beta_{31}^*.
$$

One can see that a change from constrained to unconstrained problem was done by adding and deleting columns of the model matrix.
 
 The OLS estimator of $\boldsymbol\beta^*$ constrained to $\mathcal{L}_M$ is given by the following expression:
\begin{equation}\label{eq:betaM}
\widehat{\boldsymbol\beta}_M = \bold{A}_M^1 \widehat{\boldsymbol\xi}_M, \text{ where } \widehat{\boldsymbol\xi}_M = \left( \bold{Z}_{1M}^T \bold{Z}_{1M} \right)^{-1} \bold{Z}_{1M}^T \bold{y}.
\end{equation}
Note that $\bold{A}_{0M}\widehat{\boldsymbol\beta}_M = \bold{A}_{0M}\bold{A}^1_M \widehat{\boldsymbol\xi}_M = 0$ and thus indeed $\widehat{\boldsymbol\beta}_M \in  \mathcal{L}_{M}$.
We define the inclusion relation between two models $M_1$ and $M_2$ by inclusion of linear spaces 
\begin{equation}\label{eqINCL}
M_1 \subseteq M_2 \text{ denotes } \mathcal{L}_{M_1} \subseteq \mathcal{L}_{M_2}
\end{equation}
and intersection of two models $M_1$ and $M_2$ by intersection of linear spaces:
\begin{equation}\label{eqCAP}
M_1 \cap M_2 \text{ as a model defined by } \mathcal{L}_{M_1} \cap \mathcal{L}_{M_2}.
\end{equation}
A feasible model $M$ will be called a true model if $\boldsymbol\beta^* \in \mathcal{L}_{M}$. A true model with minimal size will be denoted by $T$. Observe that $T$ is unique because $\bold{X}$ is a full rank matrix. \\
\textbf{Example 1 continued}. For the illustrative example the true model $T$ is $T = ( \{ 1 \}, \{ \{ 1, 4 \}, \{ 2, 3 \} \})$. The dimensions of the considered models are $|F| = p = 5$, $|T| = 3$.

\subsection{Residual sum of squares and generalized information criterion for feasible models}
Let $\bold{H}_M = \bold{Z}_{1M} \left( \bold{Z}_{1M}^T\bold{Z}_{1M} \right)^{-1} \bold{Z}_{1M}^T$. Observe that $\bold{H}_M \bold{X}\boldsymbol\beta^* = \bold{X}\boldsymbol\beta^*$ for $ M \supseteq T$ .
We define residual sum of squares for model $M$ as $RSS_M = \| \bold{y} - \bold{X}\widehat{\boldsymbol\beta}_M \|^2$. 
 From equation~(\ref{eq:betaM}) we have:
\[
RSS_M =  \| \bold{y} - \bold{Z}_{1M}\widehat{\boldsymbol\xi}_M \|^2 =  \|(\mathbb{I} - \bold{H}_M) \bold{y}  \|^2.
\]
Let us denote:
\begin{equation}\label{eq:Delta}
\Delta_M = \boldsymbol\beta^{*T} \bold{X}^T (\mathbb{I} - \bold{H}_M) \bold{X} \boldsymbol\beta^{*}
 = \| \bold{X}\boldsymbol\beta^* - \bold{X}\boldsymbol\beta_M^* \|^2,
\end{equation}
where $\boldsymbol\beta_M^*  = \argmin_{\boldsymbol\beta \in \mathcal{L}_M}   \| \bold{X}\boldsymbol\beta^* - \bold{X}\boldsymbol\beta\|^2$. Notice that $\widehat{\boldsymbol\beta}_M \xrightarrow[]{\ P \ } \boldsymbol\beta^*_M$ with $n \rightarrow \infty$.
The following decomposition of RSS in linear models is trivial, hence we omit the proof:
\begin{proposition}
\[ 
RSS_M = \Delta_M +  2 \boldsymbol\beta^{*T} \bold{X}^T (\mathbb{I} - \bold{H}_M) \boldsymbol\epsilon +   \boldsymbol\epsilon^T  (\mathbb{I} - \bold{H}_M) \boldsymbol\epsilon.
\]
In particular for $ M \supseteq T$ 
\[
 RSS_M = \boldsymbol\epsilon^T  (\mathbb{I} - \bold{H}_M) \boldsymbol\epsilon  \sim \sigma^2 \chi^2_{n - |M|}.
\]
\end{proposition}

Therefore, the predictions for constrained problem can be obtained through projecting the observations on the space spanned by columns of the model matrix for the equivalent unconstrained problem. Hence, decompositions and asymptotic properties of residual sums of squares for feasible models are inherited from unconstrained linear models. 

Bayes Information Criterion for model $M$ is defined as:
\[
BIC_M= n \log RSS_M  + \log(n) |M|.
\] 
The goal of our method is to find the best feasible model according to BIC, taking into account that the number of feasible models grows exponentially with $p$. Since for the $k$-th factor number of possible partitions is the Bell number $\mathcal{B}(p_k)$, the number of all feasible models is $2^{p_0} \prod_{k = 1}^{l} \mathcal{B}(p_k)$. In order to significantly reduce the amount of computations, we propose a greedy backward search.

\section{DMR algorithm}\label{secALG}
In this section we introduce  DMR algorithm. Because of troublesome notations, in order to make the description of the algorithm more intuitive, we present here  a general idea of the algorithm. In particular, we give the details of step 3 of the algorithm in the Appendix~\ref{detDMR}.

 Assuming that $\bold{X}$ is of full rank the QR decomposition of the model matrix is $\bold{X} = \bold{Q}\bold{ R}$,
where $\bold{Q}$ is $n \times p$ orthogonal matrix and $\bold{R}$ is $p \times p$ upper triangular matrix. Denote minimum variance unbiased estimators of  $\boldsymbol\beta$ and $\sigma^2$ for the full model $F$ as:
\begin{equation}\label{eqz}
\widehat{\boldsymbol\beta} = \bold{R}^{-1}\bold{z} \text{ and } \widehat{\sigma}^2 = \frac{\| \bold{y} \|^2 - \| \bold{z} \|^2}{n-p} \text{, where } \bold{z} = \bold{Q}^T \bold{y}.
\end{equation}
Let us denote
\[
\widehat{\boldsymbol\beta} = [\widehat{\beta}_{jk}]_{\begin{tiny} \begin{matrix} j \in N_k \\ k \in N \end{matrix} \end{tiny}}, \ \bold{R}^{-1} = [r_{jk, st}]_{\begin{tiny} \begin{matrix} j \in N_k \\ s \in N_t \\ k,t \in N \end{matrix} \end{tiny}},
\]
then
\[
 \widehat{\beta}_{jk} =  \bold{r}_{jk}^T \bold{z}, \text{ where } j \in N_k, k \in N
\]
and $\bold{r}_{jk}$ is a row of $\bold{R}^{-1}$.

\begin{algorithm}[H]
   \caption{DMR (Delete or Merge Regressors)}
   \label{alg:DMR}
\begin{algorithmic}
   \STATE {\bfseries Input:} $y$, $X$
   \STATE {\bfseries 1. Computation of t-statistics}
   \STATE {Compute the QR decomposition of the full model matrix, obtaining matrix $\bold{R}^{-1}$, vector $\bold{z}$ and variance estimator $\widehat{\sigma}^2$ as in equation (\ref{eqz}). Calculate squared t-statistics:}
\begin{enumerate}
\item for all elementary constraints  defined in (\ref{eqHi}): 
\[
t^2_{1jk} = \frac{\widehat{\beta}^2_{jk}}{\widehat{Var}(\widehat{\beta}_{jk})} = \frac{(\bold{r}_{jk}^T \bold{z})^2}{\widehat{\sigma}^2 \| \bold{r}_{jk} \|^2}~~\text{for}~~j \in N_k \setminus \{0\},\ k \in N,
\] 
\item for all elementary constraints defined in (\ref{eqHij}): 
\[
t^2_{ijk}  = \frac{(\widehat{\beta}_{ik} - \widehat{\beta}_{jk})^2}{\widehat{Var}(\widehat{\beta}_{ik} - \widehat{\beta}_{jk})} = \frac{((\bold{r}_{ik} - \bold{r}_{jk})^T z)^2}{\widehat{\sigma}^2 \| \bold{r}_{ik} - \bold{r}_{jk} \|^2}
\]
for  $i,j \in N_k$, $i \neq j$, $k \in N \setminus \{ 0 \}$. 
\end{enumerate}
\STATE {\bfseries 2. Agglomerative clustering for factors (using complete linkage clustering)}
\STATE{ For each factor perform agglomerative clustering using $\bold{D}_k = \left[ d_{ijk} \right]_{ij}$ as dissimilarity matrix for $k \in N \setminus \{ 0 \}$:}
\begin{enumerate}
\item $d_{1jk} = d_{j1k} = t^2_{1jk}$ for $j \in N_k$,
\item $d_{ijk} = t^2_{ijk}$ for $i, j \in N_k$, $i \neq j$,
\item $d_{iik} = 0 $ for $i \in N_k$.
\end{enumerate}
We denote cutting heights obtained from the clusterings as $ \bold{h}_1^T,  \bold{h}_2^T, \ldots, \bold{h}_l^T$.
\STATE {\bfseries 3. Sorting constraints (hypotheses) according to the squared t-statistics}
\STATE{Combine vectors of cutting heights: $\bold{h} = [0, \bold{h}_0^T, \bold{h}_1^T, \ldots, \bold{h}_l^T]^T$, where $\bold{h}_0$ is vector of squared t-statistics for constraints concerning continuous variables and $0$ corresponds to the full model. Sort elements of $\bold{h}$ in increasing order and construct a corresponding $(p-1) \times p$ matrix $\bold{A}_0$ of consecutive constraints.}
\STATE {\bfseries 4. Computation of RSS using a recursive formula in a nested family of models}
\STATE{ Perform QR decomposition of the matrix $\bold{R}^{-T} \bold{A}_0^T$ obtaining the orthogonal matrix 
$\bold{W} = [\bold{w}_{1}, \ldots, \bold{w}_{p-1}].$ Set $\text{RSS}_{M_0} = \| \bold{y} \|^2 - \| \bold{z} \|^2$ for a model without constraints. }
For $m = 1, \ldots, p-1$
\[
\text{RSS}_{M_m} = \text{RSS}_{M_{m-1}} + (\bold{w}_{m}^T \bold{z})^2,
\]
where ${M_m}$ denotes a model with constraints defined by $m$ first rows of  $\bold{A}_0$. The last formula is derived in the Appendix \ref{secRSS}, see equation (\ref{rsscon}). 
\STATE {\bfseries 5. Choosing the best model according to BIC}
\STATE{Calculate 
\[
\text{BIC}_{M_m} = n \log \text{RSS}_{M_m} + (p-m)\log(n)
\]
for $m = 0, \ldots, p-1$.
Selected model $\widehat{T}$ is the model minimizing BIC among models on the nested path:
\[
\widehat{T} = \argmin_{\substack{M_m \\ 0 \leq m \leq p-1}} \text{BIC}_{M_m}.
\]}
\STATE {\bfseries Output:} $\widehat{T}$
\end{algorithmic}
\end{algorithm}

 The time complexities of successive steps of DMR algorithm are $O(np^2)$ for QR decomposition in step 1, $O(p^2)$ for hierarchical clustering in step 2, $O(p^3)$ for QR decomposition used in step 4. 
The dominating operation in the described procedure is the QR decomposition of the full model matrix. Hence, the overall time complexity of DMR algorithm is $O(np^2)$. \\
\textbf{Example 1 continued}.
For the illustrative example we have:
\[
t_{110}^2 = 9.35 \ ,\ \bold{D_1} =   \left[\begin{array}{cccc}
0& t_{121}^2 & t_{131}^2 &  t_{141}^2 \\ 
  t_{121}^2 & 0 & t_{231}^2 &  t_{241}^2 \\ 
 t_{131}^2 & t_{231}^2 & 0 &  t_{341}^2 \\ 
  t_{141}^2  & t_{241}^2 &  t_{341}^2 & 0 \\ 
\end{array}\right] = \left[\begin{array}{cccc}
 0  & 8.01 & 4.52 & 0.20 \\ 
8.01 & 0  & 0.15 & 3.09 \\ 
 4.52  &0.15 & 0  & 2.91 \\ 
 0.20  & 3.09 & 2.91 & 0  \\ 
\end{array}\right],
\]
\[
\bold{h} = [0, 0.15, 0.20, 8.01, 9.33 ]^T \ ,\ \bold{A_0} =   \begin{blockarray}{ccccc}
\beta_{00} & \beta_{10} & \beta_{21} & \beta_{31} & \beta_{41} \\
\begin{block}{[ccccc]}
0 & 0 & -1 & 1 & 0  \\
0 & 0 & 0 & 0 & 1  \\
0 & 0 & 1 & 0 & 0 \\
0 & 1 & 0 & 0 & 0  \\
\end{block}
\end{blockarray}, 
\]

\[
\bold{BIC} = [28.33, 26.65, 25.36, 34.68, 39.59]^T.
\]
Observe that the selected model $\widehat{T}$ is the true model  $T$ .
The dendrogram and cutting heights for the illustrative example obtained from clustering in step 2  are shown in Figure~\ref{rysdendr}. The horizontal dashed line corresponds to the optimal  partition chosen by BIC.


\begin{figure}[!ht]
\centerline{%
\includegraphics[width=85mm, height=85mm]{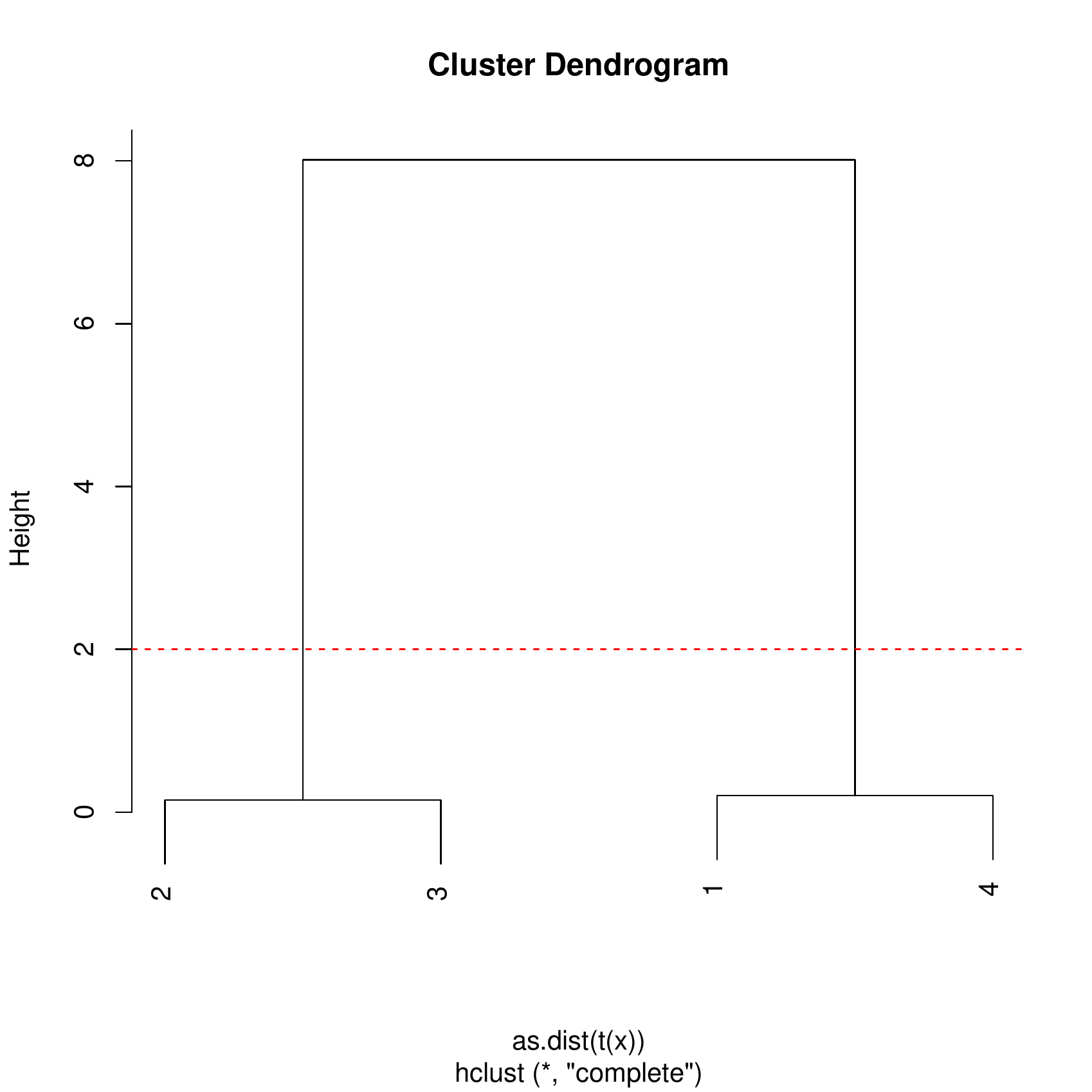}}
\caption{Dendrogram for Example 1.}\label{rysdendr}
\end{figure}

\section{Asymptotic properties of DMR algorithm}\label{secASY}
In Algorithm \ref{alg:DMR} and all the simulations and examples we assumed complete linkage in hierarchical clustering and BIC for selection in the nested family of models. The proof of consistency is more general: the linkage criterion has to be a convex combination of the minimum and maximum of the pairwise distances between clusters (see equation~\ref{eqMONO} in Appendix~\ref{appPROOF}) and generalized information criterion is used for final model selection:
\[
GIC_M= n \log RSS_M  + r_n |M|,
\]
where $r_n$ is the penalty for model size.  
Note that well known criteria AIC and BIC are special cases of GIC, if $r_n = 2$ and $r_n = \log(n)$ respectively. 

In this section we use $f_n \prec g_n$ to denote $f_n = o(g_n)$. We allow the number of predictors $p_n$ to grow monotonically with the number of observations $n$ under the condition $p_n \prec n$. 

We distinguish the following subsets of the set of all feasible models $\mathcal{M}$:
\begin{enumerate}
\item Uniquely defined model $T$, which is fixed and does not depend on sample size. We assume that the model consists of a finite number of continuous variables and a finite number of factors with finite numbers of levels. 
\item A set $\mathcal{M}_{\mathcal{V}}$ of models with one constraint imposed which is false:
\[
\mathcal{M}_{\mathcal{V}} = \{ M \subseteq F: |M| = |F| - 1 \text{ and } T \nsubseteq M \},
\]
\item A set $\mathcal{M}_\mathcal{T}$ of models with one constraint imposed which is true: 
\[
\mathcal{M}_\mathcal{T} = \{ M \subseteq F: |M| = |F| - 1 \text{ and } T \subseteq M \}.
\]
\end{enumerate}
We denote:
\begin{equation}\label{eq:DD}
\Delta = \min_{M \in \mathcal{M}_{\mathcal{V}}} \Delta_M,
\end{equation}
where $\Delta_M$ was defined in equation (\ref{eq:Delta}).
Let us notice that from equation (\ref{eq:betaM}) we get
\[
\text{Var}\left( \widehat{\boldsymbol\beta}_M \right) = \bold{A}^1_M \text{Var} \left( \widehat{\boldsymbol\xi}_M \right) \bold{A}^{1T}_M = \bold{A}^1_M  \left(\bold{A}^{1T}_M \bold{X}^T \bold{X} \bold{A}^1_M   \right)^{-1} \bold{A}^{1T}_M .
\]
Then
\[
\text{Var}\left( \sqrt{n} \left( \widehat{\boldsymbol\beta}_M - \boldsymbol\beta^* \right) \right)= n \bold{A}_M^1 \left( \bold{A}_M^{1T} \bold{X}^T \bold{X} \bold{A}_M^{1} \right)^{-1} \bold{A}_M^{1T}.
\] 
Additionally, for finite $p$, independent of $n$, if  $\frac{1}{n} \bold{X}^T \bold{X} \rightarrow \boldsymbol\Sigma > 0$ then
\[
\text{Var}\left( \sqrt{n} \left( \widehat{\boldsymbol\beta}_M - \boldsymbol\beta^* \right) \right)  \rightarrow \boldsymbol\Sigma_M = \bold{A}_M^1 \left( \bold{A}_M^{1T} \boldsymbol\Sigma \bold{A}_M^{1} \right)^{-1} \bold{A}_M^{1T}.
\] 
\begin{theorem}\label{theoremASY}

Assume that $\bold{X}$ is of full rank and $p_n \prec r_n \prec \min (n, \Delta)$. Let $\widehat{T}$ be the model selected by DMR, where linkage criterion for hierarchical clustering  is a convex combination of minimum and maximum of the pairwise distances between clusters. Then
\begin{enumerate}
\item[(a)]
$ \lim_{n \rightarrow \infty} \mathbb{P}(\widehat{T} = T) = 1$,
\item[(b)]
$\sqrt{n} \left(  \widehat{\boldsymbol\beta}_{\widehat{T}} - \boldsymbol\beta^* \right) \xrightarrow[]{\ d \ } \mathcal{N}(\bold{0}, \sigma^2 \boldsymbol\Sigma_T)$ if additionally $p$ is finite, independent of $n$ and $\frac{1}{n} \bold{X}^T \bold{X} \rightarrow \boldsymbol\Sigma > 0$.
\end{enumerate}
\end{theorem}
Proof can be found in the Appendix \ref{appPROOF}.

\section{Numerical experiments}\label{secSIMUL}

All experiments were performed using functions implemented in \texttt{R} package called \texttt{DMR}, which is available at the \texttt{CRAN} repository. The main function in the package is called \texttt{DMR} and implements the DMR algorithm with an optional method of hierarchical clustering (default is complete) and a value of $r_n$ in GIC (default is $\log(n)$). The package also contains other functions that are modifications of the DMR algorithm, such as stepDMR which assumes recalculation of t-statistics after accepting every new elementary constraint and DMR4glm which can be used for model selection in generalized linear models. 

We compared 2 groups of algorithms. The first one contains 3 stepwise procedures stepBIC, ffs BIC and DMR. The second group are 2 Lasso-based methods: CAS-ANOVA and gvcm. Procedure stepBIC is implemented in the function \texttt{stepAIC} in \texttt{R} package \texttt{MASS} and  does not perform factor partitions but either deletes or keeps any of categorical predictors. A factor forward stepwise procedure (ffs BIC), implemented in \texttt{R} package \texttt{gvcm.cat} is similar to DMR  but differs in the search direction (DMR is backward and ffs BIC is forward) and in the criterion of selection of the best step (DMR uses t-statistics calculated only once and hierarchical clustering and ffs BIC recalculates criterion in every step). For DMR the complete linkage method of clustering and BIC were used. Algorithm gvcm is implemented in  \texttt{R} package \texttt{gvcm.cat} where by default there are no adaptive weights and crossvalidation is used for choosing the $\lambda$ parameter. We used adaptive weights and BIC criterion for choosing the tuning parameter since we got better results then. Implementation of CAS-ANOVA can be found on the website \url{http://www4.stat.ncsu.edu/~bondell/Software/ CasANOVA/CasANOVA.R}. Here the default BIC was used for choosing the $\lambda$ parameter making all the methods dependent on the same criterion of choosing the tuning parameters. Adaptive weights are also default in CAS-ANOVA. When using the two Lasso-based algorithms we found difficult the selection of the $\lambda$ grid. In all the experiments we tried different grids: the default ones and ours both on linear and logarithmic scales  presenting only the best results. 

We describe three simulation experiments. In Section~\ref{subsecBONDELL} results regarding an experiment constructed in the same way as in \cite{bondell} is presented.  The model consists of three factors and no continuous variables. 
As a continuation, simulations based on data containing one factor and eight correlated continuous predictors were carried out, the results can be found in Section~\ref{subsecSKOR}. 
In Section~\ref{subsecGLM} we summarize the results of an experiment regarding generalized linear models. In this experiment only 4 algorithms were compared since CAS-ANOVA applies only to normal distribution.

In Section~\ref{secMEAS} we introduce measures of performance which are generalizations of popular true positive rate and false discovery rate on categorical predictors. We call them $TPR^*$ and $FDR^*$. In comparison to generalizations introduced in \cite{tutz} and \cite{bondell}, which we call $TPR$ and $FDR$, our measures don't diminish the influence of continuous predictors and factors with a small number of levels. 
Hence, for evaluation of the model selection methods we used following criteria:  true model (TM) represents the percentage of times the procedure chose the entirely correct model. Correct factors (CF) represents the percentage of times the non-significant factors were eliminated and the true factor was kept. $1 - $TPR, FDR, $1 - $TPR$^*$ and FDR$^*$ are averaged errors made by selectors described in Section~\ref{secMEAS}. MSEP stands for mean squared error of prediction for new data and MD is mean dimension of the selected model, both with standard deviations.

The last Section~\ref{subsecCARS} refers to two real data examples where  barley yield  and prices of apartments in Munich were modeled.

\subsection{Measures of performance}\label{secMEAS}
When performing simulations, results are usually compared to the underlying truth.
 Traditionally, for model selection with only continuous predictors measures such as true positive rate (TPR) or false discovery rate (FDR) are used. In the literature (\cite{tutz}, \cite{bondell}) their generalization to both continuous and categorical predictors can be found.

Let us consider sets of elementary constraints corresponding to the true and selected models determined by sets of indexes: 
\[
\mathcal{B} = \{ (i, j, k): \ i \neq j, i,j \in N_k, k \in N \setminus \{0\}, \ (\boldsymbol\beta^*)_{ik} = (\boldsymbol\beta^*)_{jk} \} 
\]
\[
\cup \{ (j, k): j \in N_k, k \in N, (\boldsymbol\beta^*)_{jk} = 0 \}
\]
 and 
\[
\widehat{\mathcal{B}} = \{ (i, j, k): \ i \neq j, i,j \in N_k, k \in N \setminus \{0\}, \ (\widehat{\boldsymbol\beta}_{\widehat{T}})_{ik} = (\widehat{\boldsymbol\beta}_{\widehat{T}})_{jk} \} 
\]
\[
\cup \{ (j, k): j \in N_k, k \in N, (\widehat{\boldsymbol\beta}_{\widehat{T}})_{jk} = 0 \}.
\]
True positive rate is the proportion of true differences which were correctly identified to all true differences, meaning ratio of the number of true elementary constraints which were found by the selector to the number of all true elementary constraints
$
TPR_{} = | \mathcal{B}\cap \widehat{\mathcal{B}} | / |\mathcal{B}|.
$
False discovery rate is the proportion of false differences which were classified as true to all differences classified as true, meaning ratio of the number of false elementary constraints which were accepted by the selector to the number of all accepted elementary constraints
$
FDR_{} = 1- | \mathcal{B} \cap \widehat{\mathcal{B}} | / |\widehat{\mathcal{B}}|.
$

However, measures defined in this way diminish the influence of the continuous variables and factors with a small number of levels. As an example, consider a model with 5 continuous predictors and one factor with 5 levels. Then the number of parameters for continuous predictors is 5 and the number of possible elementary constraints equals 5. The number of parameters for the categorical variable is also 5, whereas the number of possible elementary constraints is $\binom{5}{2} = 10$.

We introduce a different generalization of traditional performance measures using dimensions of linear spaces which define the true and selected models. 
We consider two models: true model $T$ and selected model $\widehat{T}$. 

We define true positive rate coefficient as
$
TPR^* = |T \cap \widehat{T}| / |T|
$
and false discovery rate coefficient as
$
FDR^* = 1 - |T \cap \widehat{T}| / |\widehat{T}|,
$
where $T \cap \widehat{T}$ is defined according to equation (\ref{eqCAP}). This generalization is more fair since the influence of every parameter on the coefficients is equal.
In the article the attention is focused on values: $1-TPR^*$ and $FDR^*$, which correspond to the errors made by selector.


\subsection{Experiment 1}\label{subsecBONDELL}
The layout of this experiment is the same as in \cite{bondell}. Despite using different $\lambda$ grids, we weren't able to obtain as good results for CAS-ANOVA as in the original paper. However, the results for DMR are much better in terms of TM than those for CAS-ANOVA originally reported in  \cite{bondell}.
The experimental model consists of three factors having eight, four and three levels, respectively. The true model is $T= ( P_{1},  P_{2} ,  P_{3}  )$, where
\[
P_{1} = ( \left\{ 1,2 \right\}, \left\{3,4,5,6\right\}, \left\{ 7,8 \right\}),\ P_{2} = \left\{ 1,2,3,4 \right\},\ P_{3} = \left\{1,2,3\right\}.
\] 
The response $\bold{y}$ was generated using the true model:
\[
\bold{y} = \boldsymbol\mu + \boldsymbol\varepsilon, \ \boldsymbol\varepsilon \sim \mathcal{N} (\bold{0}, \mathbb{I}),
\]
where
\[
\begin{split}
\boldsymbol\mu = &  \bold{1}_n \beta_{00}^*  +  \bold{X}_1 \boldsymbol\beta_1^* + \bold{X}_2 \boldsymbol\beta_2^* + \bold{X}_3 \boldsymbol\beta_3^* \\
 = & \bold{1}_n \cdot 2 + \bold{X}_1 (0,-3,-3,-3,-3,-2,-2)^T +  \bold{X}_2 (0,0,0)^T + \bold{X}_3 (0,0)^T.
\end{split}
\]
 A balanced design was used with $c$ observations for each combination of factor levels, which gives $n=96 \cdot c$, $c=1,2,4$. 

The data was generated 1000 times. The best results for $\lambda_{\text{CAS-ANOVA}} = (0.1, 0.2,\ldots,3)^T$ and $\lambda_{\text{gvcm}} = (0.01,0.02,\ldots,3)^T$ together with outcomes from other methods are summarized in Table~\ref{tabbondell}.
The results of Experiment 1 indicate that DMR and ffs BIC algorithms performed almost twice better than CAS-ANOVA and gvcm in terms of choosing the true model. 
Our procedure and ffs BIC chose approximately smaller models with dimension closer to the dimension of the underlying true model, whose number of parameters is three. 
There were no significant differences between mean squared errors of prediction for all considered algorithms. 
The main conclusion, that DMR and ffs BIC procedures choose models which are smaller and closer to the proper one, is supported by the obtained values of 1 - TPR$^*$ and FDR$^*$.

\begin{table}[!ht]
\begin{center}
\caption{Results of the simulation study, Experiment 1. }\label{tabbondell}
\begin{tabular}{crcccccccc}
  \hline
n & Algorithm & TM(\%) & CF(\%) & 1-TPR$_{}$ & FDR$_{}$ & 1-TPR$^*$ & FDR$^*$ & MSEP$\pm$sd & MD$\pm$sd \\ 
  \hline
96 & DMR & 44 & 73 & 0.05 & 0.09 & 0.1 & 0.19 & 1.091$\pm$.179  & 3.4$\pm$.7 \\ 
 & ffs BIC & 42 & 73 & 0.04 & 0.09 & 0.1 & 0.2 & 1.091$\pm$.179 & 3.5$\pm$.7 \\ 
 & CAS-ANOVA & 17 & 83 & 0.04 & 0.14 & 0.06 & 0.33 & 1.104$\pm$.175 & 5.5$\pm$ 1.7\\ 
  & gvcm & 11 & 49 & 0.08 & 0.15 & 0.1 & 0.34 & 1.118$\pm$.179 & 4.5$\pm$1.6\\ 
  & stepBIC & 0 & 97 & 0 & 0.29 & 0 & 0.63 & 1.089$\pm$.171 & 8.1$\pm$.4\\ 
\hline
 192 & DMR & 66 & 82 & 0.01 & 0.05 & 0.02 & 0.1 & 1.036$\pm$.11 & 3.3$\pm$.6\\ 
 &  ffs BIC & 67 & 83 & 0.01 & 0.05 & 0.02 & 0.1 & 1.035$\pm$.11 & 3.3$\pm$.5\\ 
 & CAS-ANOVA & 33 & 93 & 0 & 0.09 & 0.01 & 0.24& 1.049$\pm$.109 & 4.9$\pm$1.3\\ 
 & gvcm & 27 & 60 & 0.01 & 0.11 & 0.02 & 0.27 & 1.049$\pm$.11 & 4.3$\pm$1.2\\ 
 & stepBIC & 0 & 99 & 0 & 0.29 & 0 & 0.63 & 1.046$\pm$.109 & 8$\pm$.2\\
\hline
 384 & DMR & 80 & 89 & 0 & 0.03 & 0 & 0.05 & 1.013$\pm$.074 & 3.2$\pm$.4 \\ 
& ffs BIC  & 79 & 89 & 0 & 0.03 & 0 & 0.05 & 1.013$\pm$.074 & 3.2$\pm$.4 \\ 
&  CAS-ANOVA   & 50 & 97 & 0 & 0.06 & 0 & 0.17 & 1.022$\pm$.074 & 4.2$\pm$1.2 \\ 
& gvcm & 49 & 77 & 0 & 0.06 & 0 & 0.16 & 1.02$\pm$.074  & 3.8$\pm$1 \\ 
&  stepBIC & 0 & 100 & 0 & 0.29 & 0 & 0.63 & 1.022$\pm$.074 & 8$\pm$.1 \\ 
   \hline
\end{tabular}
\end{center}
\end{table}

An exemplary run of DMR algorithm is shown in Figure~\ref{rysden}. The horizontal dotted line indicates the cutting height for the best model chosen by BIC.

\begin{figure}[!ht]
\centerline{%
\includegraphics[width=150mm, height=55mm]{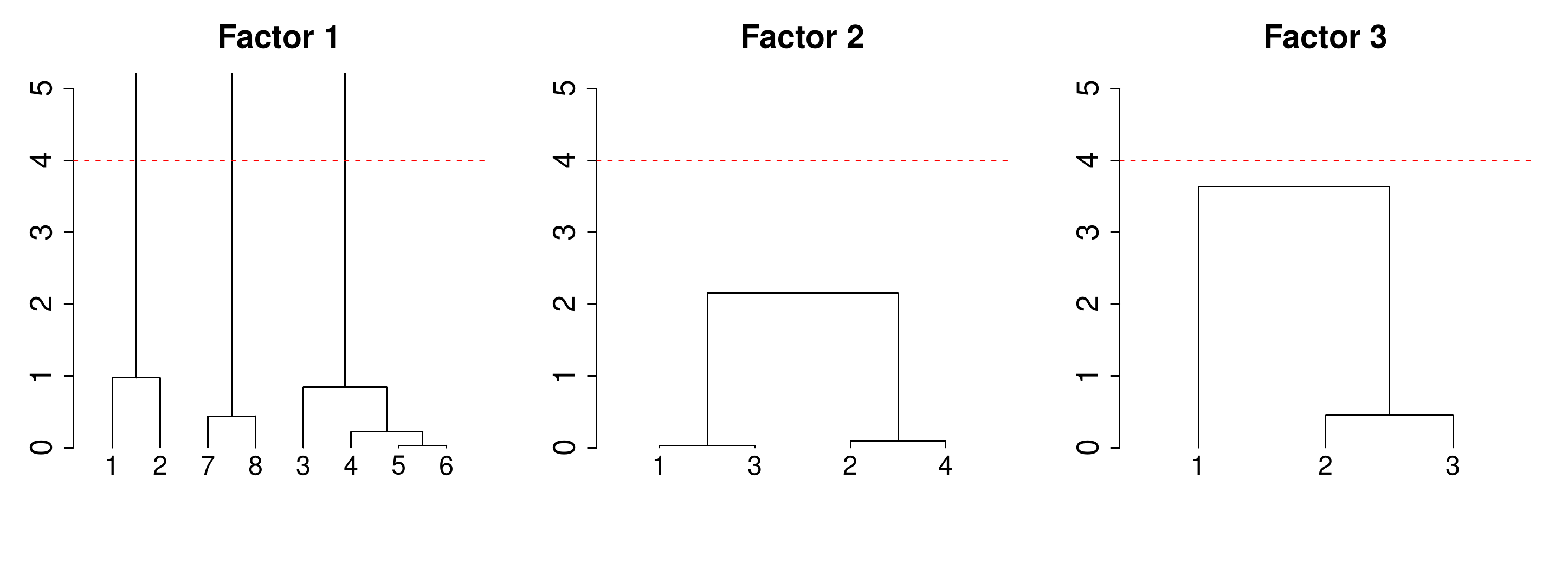}}
\caption{An examplary run of DMR algorithm for Experiment 1.}\label{rysden}
\end{figure}

In Table~\ref{tabct} the computation times of the algorithms are summarized. All values are divided by the computation time of \texttt{lm.fit} function, which fits the linear model with the use of QR decomposition of the model matrix. 

The results for CAS-ANOVA and gvcm are given for only one value of $\lambda$. By default, the searched lambda grid is of length 50 and 5001, respectively. One can see that DMR is significantly faster than ffs BIC, CAS-ANOVA  and gvcm.

\begin{table}[!ht]
\begin{center}
\caption{Computation times divided by the computation time of \texttt{lm.fit}, results obtained using \texttt{system.time} function.}\label{tabct}
\begin{tabular}{ccccccc}
  \hline
c & n & DMR & ffs BIC & CASANOVA & gvcm & stepBIC\\ 
  \hline
1 & 96 & 87 & 883& 234 & 250 & 71\\ 
  4 & 384 & 36 & 526 & 89 & 245 & 31\\ 
  20 & 1920 & 19 & 394 & 21 & 739 & 16\\ 
   \hline
\end{tabular}
\end{center}
\end{table}

\subsection{Experiment 2}\label{subsecSKOR}
In the second experiment a model containing not only categorical predictors, but also continuous variables is considered. The response $\bold{y}$ was generated from the model
with one factor with eight levels and eight continuous variables:
\[
\begin{split}
\bold{y} = & \bold{V}_0 \boldsymbol\alpha_0 + \bold{V}_1 \boldsymbol\alpha_1 + \boldsymbol\varepsilon \\
 = &\bold{V}_0 (1,0,1,0,1,0,1,0)^T + \bold{V}_1 (0,0,-2,-2,-2,-2,4,4)^T + \boldsymbol\varepsilon,
\end{split}
\] 
where $\bold{V}_0$ was generated from the multivariate normal distribution with autoregressive correlation structure with $\rho = 0.8$. The first $2 \cdot 16 \cdot c$ rows were generated using mean vector $(1,1,0,0,0,0,0,0)^T$, then $4 \cdot 16 \cdot c$ observations using mean vector $(0,0,1,1,1,1,0,0)^T$ and the last $2 \cdot 16 \cdot c$ observations using mean vector $(0,0,0,0,0,0,1,1)^T$, according to the underlying true partition of the factor. $c = 1,2,4$, hence $n = 128 \cdot c$. 
$\bold{V}_1$ is a matrix of dummy variables encoding levels of the factor and $\boldsymbol\varepsilon$ was generated from zero-mean normal distribution, $\boldsymbol\varepsilon \sim \mathcal{N}(\bold{0}, \mathbb{I})$.  The data was generated 1000 times.

The best results for $\lambda_{\text{CAS-ANOVA}} = (0.1, 0.2,\ldots,3)^T$ and $\lambda_{\text{gvcm}} = (0.01,0.02,\ldots,5)^T$ together with outcomes from other methods are summarized in Table~\ref{tabskor}.
Despite the fact that additional continuous variables were correlated, the obtained results show a considerable advantage of DMR algorithm over other methods. 


\begin{table}[!ht]
\begin{center}
\caption{Results of the simulation study, Experiment 2.}\label{tabskor}
\begin{tabular}{crccccccc}
  \hline
n & Algorithm & TM(\%) & 1-TPR & FDR & 1-TPR$^*$ & FDR$^*$ & MSEP$\pm$sd & MD$\pm$sd \\ 
  \hline
128 & DMR & 68   & 0  & 0.03 & 0  & 0.05 & 1.076$\pm$.148  & 7.4$\pm$.6  \\ 
&  ffs BIC & 60   & 0.01 & 0.04 & 0.01 & 0.06 & 1.081$\pm$.15  & 7.3$\pm$.8  \\ 
&  CAS-ANOVA  & 17   & 0  & 0.13 & 0  & 0.21 & 1.11$\pm$.153 & 9.9$\pm$1.6  \\ 
&  gvcm & 12   & 0.02 & 0.11 & 0.01 & 0.23 & 1.113$\pm$.154 & 8.2$\pm$1.5  \\ 
&  stepBIC & 0  & 0  & 0.25 & 0  & 0.42 & 1.101$\pm$.148 & 12.1$\pm$.4  \\ 
\hline
 256 & DMR & 78  & 0  & 0.02 & 0  & 0.03 & 1.033$\pm$.093   & 7.2$\pm$.5  \\ 
&  ffs BIC & 54   & 0  & 0.03 & 0  & 0.07 & 1.034$\pm$.093  & 7.4$\pm$.8  \\ 
&  CAS-ANOVA    & 27  & 0  & 0.1 & 0  & 0.16 & 1.049$\pm$ .096 & 9.2$\pm$1.4  \\ 
&  gvcm  & 24   & 0  & 0.07 & 0  & 0.17 & 1.047$\pm$.096  & 7.5$\pm$1.3  \\ 
&  stepBIC & 0  & 0  & 0.25 & 0  & 0.42 & 1.049$\pm$.095  & 12.1$\pm$.3  \\
\hline
512  & DMR & 88  & 0  & 0.01 & 0  & 0.02 & 1.015$\pm$.066  & 7.1$\pm$.4  \\ 
& ffs BIC & 85  & 0  & 0.01 & 0  & 0.02 & 1.016$\pm$.066  & 6.9$\pm$.6  \\ 
& CAS-ANOVA    & 46  & 0  & 0.06 & 0  & 0.1 & 1.024$\pm$.067  & 8.4$\pm$1.2  \\ 
&  gvcm & 35   & 0  & 0.05 & 0  & 0.12 & 1.021$\pm$.067  & 7$\pm$1.1  \\ 
&  stepBIC & 0   & 0  & 0.25 & 0  & 0.42 & 1.023$\pm$.067  & 12$\pm$.2  \\ 

   \hline
\end{tabular}
\end{center}
\end{table}

\subsection{Experiment 3}\label{subsecGLM}
Simultaneous deleting continuous variables and merging levels of factors can also be considered in the framework of generalized linear models. The problem has already been discussed in \cite{tutzGLM}, where $L_1$ regularization was used. 
After replacing squared t-statistics with squared Wald's statistics, DMR algorithm can be easily 
modified 
to generalized linear models. Simulation results for DMR algorithm for logistic regression are presented below. Let us consider a logistic regression model whose linear part consists of three factors defined as in Experiment 1. The response $\bold{y}$ was independently sampled from binomial distribution:
\[
y_i \sim B \left(1, \frac{\exp(\mu_i)}{1 + \exp(\mu_i)} \right), \ i = 1,\ldots,n,
\] 
where $\mu_i$ are elements of $\boldsymbol\mu$ defined as in Experiment 1,  $\boldsymbol\mu = (\mu_1, \ldots, \mu_n)^T$ and $n = 96 \cdot c$ for $c = 1,2,4,8$.

The results  of the experiment are summarized in Table~\ref{tabglm}. The best outcomes for gvcm, presented in the table, were obtained for $\lambda$ grids $\lambda_{\text{gvcm}} = (0.01, 0.02, \ldots, 5)^T$. Again, DMR and ffs BIC show considerable advantage over other model selection methods.


\begin{table}[!ht]
\begin{center}
\caption{Results of the simulation study for logistic regression, Experiment 3.}\label{tabglm}
\begin{tabular}{crcccccccc}
  \hline
n & Algorithm & TM & CF & 1-TPR & FDR & 1-TPR$^*$ & FDR$^*$ & MSEP$\pm$sd & MD$\pm$sd \\ 
  \hline
96 & DMR & 6 & 62 & 0.21 & 0.15 & 0.38 & 0.35 & 0.304$\pm$.049 & 3.1$\pm$1.2 \\ 
 &ffs BIC & 7 & 72 & 0.21 & 0.14 & 0.37 & 0.35 & 0.302$\pm$.049  & 3.1$\pm$.8 \\ 
 & gvcm & 0 & 21 & 0.18 & 0.32 & 0.27 & 0.61 & 0.317$\pm$.062 & 6.4$\pm$2.9 \\ 
 & stepBIC & 0 & 96 & 0.00 & 0.29 & 0.00 & 0.63 & 0.299$\pm$.049  & 8$\pm$.6 \\ 
\hline
 192 &DMR & 25 & 81 & 0.16 & 0.09 & 0.25 & 0.23 & 0.296$\pm$.036  & 3$\pm$.7 \\ 
&  ffs BIC & 21 & 82 & 0.17 & 0.10 & 0.28 & 0.26 & 0.293$\pm$.034 & 3$\pm$.7 \\ 
&  gvcm & 1 & 26 & 0.15 & 0.26 & 0.19 & 0.52 & 0.296$\pm$.038 & 5.8$\pm$2.6 \\ 
&  stepBIC & 0 & 99 & 0.00 & 0.29 & 0.00 & 0.63 & 0.291$\pm$.034  & 8$\pm$.2 \\ 
\hline
 384 &DMR & 55 & 88 & 0.06 & 0.06 & 0.12 & 0.14 & 0.29$\pm$.023 & 3.1$\pm$.5 \\ 
&   ffs BIC & 51 & 88 & 0.06 & 0.06 & 0.12 & 0.16 & 0.29$\pm$.023 & 3.2$\pm$.5 \\ 
& gvcm & 6 & 37 & 0.08 & 0.20 & 0.10 & 0.43 & 0.289$\pm$.022 & 5.5$\pm$2.5 \\ 
&  stepBIC & 0 & 100 & 0.00 & 0.29 & 0.00 & 0.63 & 0.289$\pm$.022 & 8$\pm$.2 \\ 
\hline
 768 &DMR & 79 & 92 & 0.01 & 0.03 & 0.03 & 0.07 & 0.29$\pm$.016  & 3.1$\pm$.4 \\ 
&  ffs BIC & 79 & 92 & 0.01 & 0.03 & 0.03 & 0.06 & 0.29$\pm$.016  & 3.1$\pm$.4 \\ 
&  gvcm & 20 & 48 & 0.01 & 0.16 & 0.02 & 0.36 & 0.289$\pm$.016  & 5.2$\pm$2.2 \\ 
&  stepBIC & 0 & 100 & 0.00 & 0.29 & 0.00 & 0.63 & 0.29$\pm$.016  & 8$\pm$.1 \\ 

   \hline
\end{tabular}
\end{center}
\end{table}

In Table~\ref{tabctglm} the computation times of the algorithms are summarized. All values are divided by the computation time of \texttt{glm.fit} function.
The results for gvcm are given for only one value of $\lambda$, while by default the searched lambda grid is of length 5001. DMR is again significantly faster than ffs BIC and gvcm.

\begin{table}[!ht]
\begin{center}
\caption{Computation times divided by the computation time of \texttt{glm.fit}, results obtained using \texttt{system.time} function.}\label{tabctglm}
\begin{tabular}{cccccc}
\hline
c & n & DMR & ffs BIC &  gvcm & stepBIC\\ 
\hline
1 & 96 & 103 & 399 & 101& 40 \\ 
4 & 384 & 68 & 398 & 74 & 28 \\ 
20 & 1920 & 49 & 377 & 101 & 23 \\ 
\hline
\end{tabular}
\end{center}
\end{table}

\subsection{Real data examples}\label{subsecCARS}
\paragraph{Example 1: Barley.} The data set  \texttt{barley} from \texttt{R} library \texttt{lattice} has already been discussed in the literature, for example in \cite{bondell}. The response is the barley yield for each of 5 varieties (Svansota,  Manchuria, Velvet,    Peatland  and Trebi) at 6 experimental farms in Minnesota for each year of the years 1931 and 1932 giving a total of 60 observations. The characteristics of the chosen models using different algorithms are presented in Table \ref{Tab}.  The results for the full model which is least squares estimator with all variables were given as a benchmark.
For the two Lasso-based algorithms we find difficult the selection of the $\lambda$ grid. Therefore, the results for CAS-ANOVA are given for two different grids: the first one chosen so that the chosen model was the same as the one described in \cite{bondell}, $\bold{\lambda}_1 = (25, 25.01, 25.02, \ldots, 35)^T$, and the second wider superset of the first one, $\bold{\lambda}_1 = (0.1, 0.2, 0.3,\ldots, 35)^T$. We used $\bold{\lambda}_2$ grid also for gvcm.

The results show that stepwise methods give smaller models with smaller BIC values than the Lasso-based methods. The additional advantage of DMR and ffs BIC is lack of a troublesome tuning parameter.

\begin{table}[H]
\centering
\caption{Characteristics of the chosen models for Barley data set. \label{Tab}}
\begin{tabular}{rcccc}
  \hline
algorithm & model dim & $R^2$ & adj. $R^2$  & BIC \\ 
  \hline
full model & 11 & .68 & .61 & 416 \\ 
  stepBIC & 11 & .68 & .61 & 416 \\ 
  CAS-ANOVA $\bold{\lambda}_2$ & 9 & .66 & .61 & 411 \\  
gvcm $\bold{\lambda}_2$ & 7 & .66 & .6 & 403 \\ 
 CAS-ANOVA $\bold{\lambda}_1$ & 6 & .61 & .58 & 407 \\ 
  ffs BIC & 5 & .64 & .61 & 399 \\ 
  DMR & 5 & .64 & .61 & 399 \\ 
   \hline
\end{tabular}
\end{table}
\paragraph{Example 2: Miete.} The data set \texttt{miete03} comes from \url{http://www.statistik.lmu.de/service/datenarchiv}. The data consists of 2053 households interviewed for the Munich rent standard 2003. The response is monthly rent per square meter in Euros.  8 categorical and 3 continuous variables  give 36 and 4 (including the intercept) parameters. The data is described in detail in \cite{tutz}. 

Model selection was performed using five methods: DMR, ffs BIC, CAS-ANOVA, gvcm and stepBIC. Characteristics of the chosen models are shown in Table~\ref{tabcars} with results for the full model added for comparison. 



\begin{table}[!ht]
\begin{center}
\caption {Characteristics of the chosen models for Miete data set.}\label{tabcars}
\begin{tabular}{lcccc}
 \hline
 Selection &  Model & $R^2$ & adj.$R^2$   & BIC   \\
method & dimension &  & & \\
\hline
  Full model & 40 & .94 & .94 & 23037\\
 CAS-ANOVA & 31 & .94 & .94 & 22972 \\
 gvcm & 26 & .94 & .94 & 22933 \\
 DMR & 12 & .94 & .94 & 22833 \\
 stepBIC & 11 & .94 & .94 & 22847 \\
\hline
\end{tabular}
\end{center}
\end{table}

The reason of lack of results for ffs BIC in the part of Table~\ref{tabcars} is that the algorithm required to allocate too much memory (factor urban district has 25 levels).

We can conclude that DMR procedure and ffs BIC chose much better models than other compared methods in terms of BIC. However, DMR method can be applied to problems with larger number of parameters.

\section{Discussion}
We propose the DMR method which combines deleting continuous variables and merging levels of factors in linear models. DMR 
relies on ordering of elementary constraints using squared t-statistics and choosing the best model according to BIC in 
the nested family of models. A slightly modified version of the DMR algorithm can be applied to generalized linear models.

We proved that DMR is a consistent model selection method. The main advantage of our theorem over the analogous one for 
the Lasso based methods (CAS-ANOVA, gvcm) is that we allow that the number of predictors grows to infinity.

We show in simulations that DMR and ffs BIC are more accurate than the Lasso-based methods. However, DMR is much faster 
and less memory demanding
in comparison to ffs BIC. Our results  are not exceptional in comparison to others in the literature. In Example 1 in 
\cite{ZouLi} a similar
simulation setup to our Experiment 1, $n=96$, has been considered. The adaptive Lasso method (denoted there as one-step LOG) 
was outperformed by exhaustive BIC with 66 to 73 percent of true model selection accuracy. We repeated the 
simulations and got similar results with 76 percent for the Zheng-Loh algorithm (described in \cite{loh}), which is DMR with just 
continuous variables. Thus, in the Zou and Li experiment the advantage of the Zheng-Loh algorithm over the adaptive Lasso 
is not as large as in our work, but Zou and Li used a better local linear approximations (LLA) of the penalty function in 
the adaptive Lasso
implementation. Recall that both CAS-ANOVA and gvcm employ the local quadratic approximation (LQA) of the penalty 
function.

The superiority of DMR over the Lasso based methods in our experiments not only comes from weakness of LQA used in the 
adaptive Lasso implementation.
Greedy subset selection methods similar to the Zheng-Loh algorithm have been proposed many times.
Recently, in \cite{PaM} a combination of screening of predictors by the Lasso with the Zheng-Loh greedy 
selection for high-dimensional linear models has been proposed. The authors  showed both theoretically and experimentally that such combination is 
competitive to the Multi-stage Convex Relaxation described in \cite{zhang2010}, which is least squares with capped $l_1$ 
penalty implemented via LLA.

\


\appendix

\section{Regular form of constraint matrix}\label{sec:regular}
We say that $\bold{A}_{0M}$ is in regular form if it can be complemented to  $\bold{A}_M$ so that:
\begin{equation}\label{eqCONSTR}
\bold{A}_M = \left[\begin{array}{c}
\bold{A}_{1M} \\ \hline
\bold{A}_{0M} \\
\end{array}\right] = 
\left[\begin{array}{cc}
\mathds{I} &0 \\ \hline
\bold{B}_M &\mathds{I}
\end{array}\right],
\end{equation}
where $\bold{B}_M$ is a matrix consisting of $0, -1,1$. Then, using Schur complement we get:
\begin{equation}\label{eqA1}
\bold{A}_M^{-1} = \left[\begin{array}{ccc}
\mathds{I} & \vline & 0 \\
- \bold{B}_M & \vline & \mathds{I} \\
\end{array}\right] = \left[\begin{array}{ccc}
\bold{A}_M^1 & \vline & \bold{A}_M^0 \\
\end{array}\right].
\end{equation}
Constraint matrix in regular form can always be obtained by a proper permutation of model's parameters.
 Let us denote clusters in each partition: $P_{Mk} = \left( C_{Mik} \right)_{i=1}^{i_k}$, where $i_k$ is the number of clusters, $k \in N \setminus \{0 \}$  and minimal elements in each cluster as $j_{Mik} = \min \{j \in C_{Mik}\}$. Let $P_{M0}$ denote the set of continuous variables in the model.
 Sort model's parameters in the following order:
\begin{enumerate}
\item $\beta_{00}$,
\item $\beta_{j0}$: $j \in P_{M0} \setminus \{ 0\}$,
\item $\beta_{j_{Mik}k}$ for  $i = 1,\ldots,i_k$, $i \neq 1$, $k \in N \setminus \{0 \}$,
\item $\beta_{j0}$: $j \in N_0 \setminus P_{M0}$,
\item $\beta_{jk}$, $j \in C_{Mik}\setminus  \{j_{Mik}\}$, $k \in N \setminus \{0 \}$.
\end{enumerate}
Sort columns of model matrix $\bold{X}$ in the same way as vector $\boldsymbol\beta$.

\begin{example}
\emph{ As an illustrative example consider a full model $F = ( P_{F0},  P_{F1},  P_{F2})$, where
$$
P_{F0} = \{1,2\},\  P_{F1} = \left( \{1 \},\{2 \},\{3 \}, \{4 \},\{5 \}, \{6 \},\{7 \},\{8 \}\right),\  P_{F2} = \left(\{1 \},\{2 \}, \{3 \}\right)
$$
and $p_0 = 2, p_1 = 8, p_2 =3, p=12$. We denote a feasible model with 7 elementary constraints:
$\beta_{10}=0, \ \beta_{21}=0,\ \beta_{71}=0,\ \beta_{31}=\beta_{51},\ \beta_{41}=\beta_{61},\ \beta_{41}=\beta_{81}, \ \beta_{22}=0$
as  $M=  ( P_{M0},  P_{M1},  P_{M2})$, where : 
$$
P_{M0} = \{2\}, \ P_{M1} = \left( \left\{1,2,7 \right\}, \left\{3,5\right\}, \left\{4,6,8\right\} \right), \ P_{M2} = \left( \left\{1,2\right\}, \left\{3 \right\} \right).
$$ 
Constraint matrix in regular form for model $M$, where each row corresponds to one of the 7 elementary constraints,  is:
\[
\bold{A}_{0M} = \begin{array}{cccccccccccc}
\beta_{00}& \beta_{20} & \beta_{31} & \beta_{41} & \beta_{32} & \beta_{10} & \beta_{21} & \beta_{71} & \beta_{51} & \beta_{61} & \beta_{81} & \beta_{22} \\ 
\left[\begin{array}{c}  0 \\ 0 \\ 0 \\ 0 \\ 0 \\ 0 \\ 0 \\ \end{array}\right.
&
\begin{array}{c} 0 \\ 0 \\ 0 \\ 0 \\ 0 \\ 0 \\ 0 \\ \end{array} 
&
\begin{array}{c}0 \\ 0 \\ 0 \\ -1 \\ 0 \\ 0 \\ 0 \\ \end{array} 
&
\begin{array}{c} 0 \\ 0 \\ 0 \\ 0 \\ -1 \\ -1 \\ 0 \\ \end{array} 
&
\begin{array}{c} 0 \\ 0 \\ 0 \\ 0 \\ 0 \\ 0 \\ 0 \\ \end{array}
&
\begin{array}{c} 1 \\ 0 \\ 0 \\ 0 \\ 0 \\ 0 \\ 0 \\ \end{array}
&
\begin{array}{c} 0 \\ 1 \\ 0 \\ 0 \\ 0 \\ 0 \\ 0 \\ \end{array}
&
\begin{array}{c} 0 \\ 0 \\ 1 \\ 0 \\ 0 \\ 0 \\ 0 \\ \end{array}
&
\begin{array}{c} 0 \\ 0 \\ 0 \\ 1 \\ 0 \\ 0 \\ 0 \\ \end{array}
&
\begin{array}{c} 0 \\ 0 \\ 0 \\ 0 \\ 1 \\ 0 \\ 0 \\ \end{array}
&
\begin{array}{c} 0 \\ 0 \\ 0 \\ 0 \\ 0 \\ 1 \\ 0 \\ \end{array}
&
\left.\begin{array}{c} 0 \\ 0 \\ 0 \\ 0 \\ 0 \\ 0 \\ 1 \\ \end{array}\right]
\end{array}.
\]
and after inverting matrix~$\bold{A}_M^{-1}$ is obtained
\[
\bold{A}_M^{-1} = \left[\begin{array}{c|c}
\bold{A}_M^1  & \bold{A}_M^0 \\
\end{array}\right]
=
\begin{blockarray}{ccccccccccccc}
\beta_{00}& \beta_{20} & \beta_{31} & \beta_{41} & \beta_{32} & \beta_{10} & \beta_{21} & \beta_{71} & \beta_{51} & \beta_{61} & \beta_{81} & \beta_{22} \\ 
\begin{block}{[ccccc|cccccccc]}
1 & 0 & 0 & 0 & 0 & 0 & 0 & 0 & 0 & 0 & 0 & 0 \\
0 & 1 & 0 & 0 & 0 & 0 & 0 & 0 & 0 & 0 & 0 & 0 \\
0 & 0 & 1 & 0 & 0 & 0 & 0 & 0 & 0 & 0 & 0 & 0 \\
0 & 0 & 0 & 1 & 0 & 0 & 0 & 0 & 0 & 0 & 0 & 0 \\
0 & 0 & 0 & 0 & 1 & 0 & 0 & 0 & 0 & 0 & 0 & 0 \\
0 & 0 & 0 & 0 & 0 & 1 & 0 & 0 & 0 & 0 & 0 & 0 \\
0 & 0 & 0 & 0 & 0 & 0 & 1 & 0 & 0 & 0 & 0 & 0 \\
0 & 0 & 0 & 0 & 0 & 0 & 0 & 1 & 0 & 0 & 0 & 0 \\
0 & 0 & 1 & 0 & 0 & 0 & 0 & 0 & 1 & 0 & 0 & 0 \\
0 & 0 & 0 & 1 & 0 & 0 & 0 & 0 & 0 & 1 & 0 & 0 \\
0 & 0 & 0 & 1 & 0 & 0 & 0 & 0 & 0 & 0 & 1 & 0 \\
0 & 0 & 0 & 0 & 0 & 0 & 0 & 0 & 0 & 0 & 0 & 1 \\
\end{block}
\end{blockarray}
\]
}
\end{example}

Notice that for regular constraint matrix $\bold{Z}_M$ is the full model matrix $\bold{X}$ with appropriate columns deleted or added to each other.

\section{Detailed description of step 3 of the  DMR algorithm}\label{detDMR}
Since step 3 of DMR algorithm needs complicated notations concerning hierarchical clustering, we decided to present them in the Appendix for the interested reader. In particular, we show here how the cutting heights vector $\bold{h}$ and matrix of constraints $\bold{A_0}$ are built.

Let us define vectors $\bold{a}(1,j,k)$ and $\bold{a}(i, j, k)$ (corresponding to the elementary constraints, being building blocks for $\bold{A_0}$) such that :
\begin{equation}\label{eqa1jk}
\bold{a}(1,j,k) = [a_{st}(j,k)]_{\begin{tiny} \begin{matrix}  s \in N_t \\ t \in N \end{matrix} \end{tiny}}, \ a_{st}(j,k) = \mathds{1}(s=j,t=k),
\end{equation}
\begin{equation}\label{eqaijk}
\bold{a}(i,j,k) = [a_{st}(i,j,k)]_{\begin{tiny} \begin{matrix}  s \in N_t \\ t \in N \end{matrix} \end{tiny}}, \ a_{st}(i,j,k) = \mathds{1}(s=i,t=k) - \mathds{1}(s=j,t=k).
\end{equation}

For each step $s$ of the hierarchical clustering algorithm we use the following notation for the partitions of set $\{1\} \cup N_k = \{1,2,\ldots, p_k \}$:
\[
P_{sk} = \{ C_{isk} \}_{i=1}^{p_k - s +1}, \ s = 1,\ldots, p_k.
\]
We assume complete linkage clustering:
\[
\begin{split}
d& \left( C_{i_{s+1},s+1,k} = C_{i_ssk} \cup C_{j_ssk} , C_{j_{s+1},s+1,k} = C_{o_ssk}  \right) \\
&=\max \left\{ d \left( C_{i_ssk}, C_{o_ssk}  \right), d \left( C_{j_ssk}, C_{o_ssk}  \right)  \right\}.
\end{split}
\]
Cutting heights in steps $s = 1,\ldots, p_k - 1$ are defined as:
\[
h_{sk} = \min_{i \neq j} d \left( C_{isk}, C_{jsk} \right).
\]
 Let us denote vector $\tilde{\bold{a}}_{sk}$ as an elementary constraint corresponding to cutting height $h_{sk}$, where:
 \[
\tilde{\bold{a}}_{sk} = \bold{a}(i_*, j_*, k), \ i_* = \min_{ i \in C_{i_1sk}}  i, \  j_* = \min_{j \in C_{j_1sk}} j \text{ and } (i_1, j_1) = \argmin_{i \neq j} d \left( C_{isk}, C_{jsk} \right).
\] 

Step 3 of the algorithm can be now rewritten:

 Combine vectors of cutting heights: $\bold{h} = [0, \bold{h}_0^T, \bold{h}_1^T, \ldots, \bold{h}_l^T]^T$, where $\bold{h}_0$ is vector of cutting heights for constraints concerning continuous variables and $0$ corresponds to model without constraints:
\[
 \bold{h}_k = [h_{sk}]_{s = 1}^{p_k-1}, \ k \in N \setminus \{ 0 \} \text{ and }\bold{h}_{0} = [0, t_{110}^2, t_{120}^2, \ldots, t_{1p_00}^2]^T.
\]
Sort elements of $\bold{h}$ in increasing order getting $\bold{h}_{:} = [h_{m:p}]_{m = 1}^p$ and construct $(p-1) \times p$ matrix of constraints 
\[
\bold{A}_0 = [\tilde{\bold{a}}_{2:p}, \tilde{\bold{a}}_{3:p}, \ldots, \tilde{\bold{a}}_{p:p}]^T,
\]
where $\tilde{\bold{a}}_{m:p}$ is the elementary constraint corresponding to cutting height $h_{m:p}$. Then proceed as described in Algorithm \ref{alg:DMR}.

\section{Recursive formula for RSS in a nested family of linear models}\label{secRSS}
In this section we show some implementation facts concerning the DMR algorithm. In particular an effective way  of calculation of residual sums of squares for nested models using QR decompositions is discussed.

Let us consider a linear model
with linear constraints:
\begin{equation}\label{eq:lincon}
\mathcal{L} = \left\{ \boldsymbol\beta \in \mathds{R}^p, \bold{A}_{0} \boldsymbol\beta = \bold{0} \right\},
\end{equation}
where $A_{0}$ is $(p-q) \times p$ constraint matrix.
The objective is to calculate residual sum of squares $RSS = \| \bold{y} - \bold{X}\widehat{\boldsymbol\beta} \|^2$.
 QR decomposition of the model matrix is performed
\[
\bold{X} = \bold{Q} \bold{R},
\]
where $\bold{Q}$ is $n \times p$ orthogonal matrix and $\bold{R}$ is $p \times p$ upper triangular matrix.
Let us denote $\bold{S} = \bold{R}^{-T} \bold{A}_{0}^T$, then 
\[
\bold{Q}^T \bold{y} = \bold{R} \boldsymbol\beta^* + \bold{Q}^T \boldsymbol\varepsilon \text{ and } \bold{S}^T \bold{R} \boldsymbol\beta^* = \bold{0}.
\]
After substitution 
$\bold{z} = \bold{Q}^T\bold{y}$, $\boldsymbol\gamma^* = \bold{R}\boldsymbol\beta^*$, $\boldsymbol\eta = \bold{Q}^T \boldsymbol\varepsilon$ 
we get
\begin{equation}\label{eq1}
\bold{z} = \boldsymbol\gamma^* + \boldsymbol\eta \text{ and } \bold{U}^T \bold{W}^T \boldsymbol\gamma^* = \bold{0},
\end{equation}
where $\bold{W}$ and $\bold{U}$ are respectively $p \times (p-q)$ orthogonal matrix and $(p-q) \times (p-q)$ upper triangular matrix from the QR decomposition of matrix $\bold{S}$. We have
\[
\bold{W}^T \boldsymbol\gamma^* =  \bold{U} \bold{U}^T  \bold{W}^T \boldsymbol\gamma^* = \bold{0}.
\] 
Let us denote $\bold{\overbar{W}}$ as orthogonal complement of $ \bold{W}$ to matrix with dimensions $p \times p$. 
We multiply equation~(\ref{eq1}) by $[\bold{\overbar{W}}, \bold{W}]$:
\[
[\bold{\overbar{W}}, \bold{W}]^T \bold{z} =[\bold{\overbar{W}}, \bold{W}]^T \boldsymbol\gamma^* + [\bold{\overbar{W}}, \bold{W}]^T \boldsymbol\eta \text{ and } \bold{W}^T \boldsymbol\gamma^* = 0.
\]
Therefore the OLS estimator $\widehat{\boldsymbol\gamma}$ of $\boldsymbol\gamma^*$ with constraints satisfies the following equation
\begin{equation}\label{eq2}
\left[ \begin{array}{c}
\bold{\overbar{W}}^T \bold{z} \\
0
\end{array} \right] 
= [\bold{\overbar{W}}, \bold{W}]^T \widehat{\boldsymbol\gamma}.
\end{equation}
Multiplying (\ref{eq2}) by $[\bold{\overbar{W}}, \bold{W}]$, we obtain
$
\bold{\overbar{W}}\bold{\overbar{W}}^T \bold{z}  = \widehat{\boldsymbol\gamma},
$
then
\[
(\mathbb{I} - \bold{W}  \bold{W}^T) \bold{z} = \widehat{\boldsymbol\gamma} = \bold{R} \widehat{\boldsymbol\beta}.
\]
Let $\overline{\bold{Q}}$ be an orthogonal complement of $\bold{Q}$ to matrix with dimensions $n \times n$.
The residual sum of squares for the model with linear constraints (\ref{eq:lincon}) can now be
written as
\begin{equation}\label{eq:rssqr}
\begin{split}
RSS_M = & \| \overline{\bold{Q}}^T y \|^2 + \| \bold{Q}^T(\bold{y} - \bold{X} \widehat{\boldsymbol\beta}_{M}) \|^2 = \| \bold{y}\|^2 - \| \bold{z}\|^2 + \| \bold{Q}^T \bold{y} - \bold{R} \widehat{\boldsymbol\beta}_{M}\|^2 \\
= & \|\bold{y}\|^2 - \|\bold{z}\|^2 + \| \bold{W} \bold{W}^T \bold{z}\|^2 = \|\bold{y}\|^2 - \|\bold{z}\|^2 + \| \bold{W}^T \bold{z} \|^2 \\
= & \|\bold{y}\|^2 - \|\bold{z}\|^2+ \sum_{m = 1}^{p-q} (\bold{w}_{m}^T \bold{z})^2,
\end{split}
\end{equation}
where $\bold{w}_{m}$ is the $m$-th column of $\bold{W}$.

 Denote by $(\bold{A}_0)_{m,p}, \bold{S}_{m,p}, \bold{W}_{m,p}$ and $\bold{U}_{m,p}$ submatrices of $\bold{A}_0, \bold{S}, \bold{W}$ and $\bold{U}$ respectively, obtained by retaining first $m$ rows and $p$ columns.
Let us  consider a nested family of  feasible models $M_m$, $m = 0, \ldots, p-q$ defined as 
$$
\mathcal{L}_{M_m} = \left\{ \boldsymbol\beta \in \mathds{R}^p, (\bold{A}_{0})_{m,p} \boldsymbol\beta = \bold{0} \right\}.
$$  
For $m = 0, \ldots,p-q$ we have  
$$
\bold{S}_{p,m} = \bold{W}_{p, m} \bold{U}_{m,m},
$$
 because matrix $\bold{U}_{m,m}$ is upper triangular. Since 
$
\bold{W}_{p,m}^T\bold{W}_{p,m} = \mathbb{I},
$
then $\bold{W}_{p, m} \bold{U}_{m,m}$ is QR decomposition of $\bold{S}_{p,m}$. Then from equation (\ref{eq:rssqr})
we get a recursive formula for residual sum of squares for nested models:
\begin{equation}\label{rsscon}
\begin{split}
RSS_{M_0}&=   \|\bold{y}\|^2 - \|\bold{z}\|^2, \\
RSS_{M_m} &= RSS_{M_{m-1}}+  (\bold{w}_{m}^T \bold{z} )^2 \text{ for} \ m = 1, \ldots, p-1.
\end{split}
\end{equation}
\section{Proof of Theorem 1.}\label{appPROOF} 
\subsection{Properties of orthogonal projection matrices}
For a feasible model $M$ let us define a following orthogonal projection matrix:
\[
\overline{\bold{H}}_M = \bold{\bold{X}}  (\bold{X}^T \bold{X})^{-1} \bold{A}_{0M}^T \left(\bold{A}_{0M} (\bold{X}^T\bold{X})^{-1}\bold{A}_{0M}^T \right)^{-1} \bold{A}_{0M}   (\bold{X}^T\bold{X})^{-1}\bold{X}^T.
\]
\begin{lemma}\label{lem:hhh} We have
\[
\overline{\bold{H}}_M = \bold{H}_F - \bold{H}_M.
\]
\end{lemma}

\begin{proof} For simplicity of notations in the remainder of this subsection we omit subscript $M$. Let $\bold{Z}_1 = \bold{X} \bold{A}^1$, $\bold{Z} = \bold{X} \bold{A}^{-1}$ and $\bold{Z}_0 = \bold{X}\bold{A}^0$.
We denote 
\[
 \bold{G} = \left[\begin{array}{cc}
\bold{G}_{11} & \bold{G}_{10} \\
\bold{G}_{01} & \bold{G}_{00}
\end{array}\right] = \left[\begin{array}{cc}
\bold{Z}_1^T \bold{Z}_1 & \bold{Z}_1^T \bold{Z}_0 \\
\bold{Z}_0^T \bold{Z}_1 & \bold{Z}_0^T \bold{Z}_0
\end{array}\right] =  \bold{Z}^T \bold{Z} \text{ and }
  \bold{G}^{-1} = \left[\begin{array}{cc}
\bold{G}^{11} & \bold{G}^{10} \\
\bold{G}^{01} & \bold{G}^{00}
\end{array}\right].
\]
Note that
\[
\bold{H}_F =  \bold{X}(\bold{X}^T\bold{X})^{-1} \bold{X}^T =  \bold{X}\bold{A}^{-1} (\bold{A}^{-T} \bold{X}^T\bold{X} \bold{A}^{-1})^{-1} \bold{A}^{-T} \bold{X}^T = \bold{Z} (\bold{Z}^T\bold{Z})^{-1}\bold{Z}^T.
\]
Moreover
\[
\begin{split}
 (\bold{A}_0(\bold{X}^T\bold{X})^{-1} \bold{A}_0^T )^{-1} = & \left(\bold{A}_0 \bold{A}^{-1}  \left(\bold{A}^{-T} \bold{X}^T \bold{X} \bold{A}^{-1}\right)^{-1} \bold{A}^{-T} \bold{A}_0^T \right)^{-1} \\
 = & \left[ \left[\begin{array}{cc}
\bold{0} & \mathds{I}
\end{array}\right] (\bold{Z}^T\bold{Z} )^{-1}  \left[\begin{array}{c}
\bold{0} \\
 \mathds{I}
\end{array}\right]  \right]^{-1} =( \bold{G}^{00})^{-1}
\end{split}
\]
and
\[
  \bold{A}_0 (\bold{X}^T \bold{X})^{-1}\bold{X}^T = \bold{A}_0 \bold{A}^{-1} (\bold{Z}^T \bold{Z})^{-1}\bold{A}^{-T} \bold{X}^T = \bold{A}_0 \bold{A}^{-1} (\bold{Z}^T \bold{Z})^{-1} \bold{Z}^T.
\]
Then we get from the Schur complement:
\[
\begin{split}
\bold{H}_F - \bold{H}_M &=  \bold{Z} (\bold{Z}^T\bold{Z})^{-1}\bold{Z}^T - \bold{Z}_1 (\bold{Z}_1^T\bold{Z}_1)^{-1}\bold{Z}_1^T =  \bold{Z} \bold{G}^{-1}\bold{Z}^T - \bold{Z}_1 \bold{G}^{-1}_{11}\bold{Z}_1^T \\
&= \bold{Z} \bold{G}^{-1}\bold{Z}^T - \bold{Z}_1 ( \bold{G}^{11} - \bold{G}^{10} (\bold{G}^{00})^{-1} \bold{G}^{10}  )\bold{Z}_1^T \\
&= \left[\begin{array}{cc}
\bold{Z}_1  &
\bold{Z}_0
\end{array}\right] \left[\begin{array}{cc}
\bold{G}^{11} & \bold{G}^{10} \\
\bold{G}^{01} & \bold{G}^{00}
\end{array}\right]  \left[\begin{array}{c}
\bold{Z}_{1}^T \\
\bold{Z}_{0}^T 
\end{array}\right] -  \left[\begin{array}{cc}
\bold{Z}_1  &
\bold{Z}_0
\end{array}\right] \left[\begin{array}{cc}
\bold{G}^{11} - \bold{G}^{10} (\bold{G}^{00})^{-1} \bold{G}^{10}  & \bold{0} \\
\bold{0} & \bold{0}
\end{array}\right]  \left[\begin{array}{c}
\bold{Z}_{1}^T \\
\bold{0}^T
\end{array}\right] \\
 &= \bold{Z} \left[\begin{array}{c}
\bold{G}^{10} \\
\bold{G}^{00}
\end{array}\right]
 (\bold{G}^{00})^{-1}  \left[\begin{array}{cc}
\bold{G}^{01} &
\bold{G}^{00}
\end{array}\right] \bold{Z}^T =  \bold{Z} ( \bold{Z}^T  \bold{Z} )^{-1}\left[\begin{array}{c}
\bold{0} \\
\mathds{I}
\end{array}\right]
 (\bold{G}^{00})^{-1}  \left[\begin{array}{cc}
\bold{0} &
\mathds{I}
\end{array}\right]( \bold{Z}^T  \bold{Z} )^{-1} \bold{Z}^T \\
&= \bold{\bold{X}}  (\bold{X}^T \bold{X})^{-1} \bold{A}_{M}^T \left(\bold{A}_{M} (\bold{X}^T\bold{X})^{-1}\bold{A}_{M}^T \right)^{-1} \bold{A}_{M}   (\bold{X}^T\bold{X})^{-1}\bold{X}^T = \overline{\bold{H}}_M.
\end{split}
\]
\end{proof}
\subsection{Asymptotics for residual sums of squares}
Lemmas concerning dependencies between residual sums of squares have similar construction to those described in \cite{chenchen}. Let us introduce some simplifying notations.
For two sequences of random variables $U_n$ and $V_n$ we write that $U_n <_P V_n$ if $\lim_{n \rightarrow \infty} \mathds{P} \left( U_n < V_n \right)=1$. 

Residual sum of squares for model $M$ can be decomposed into three parts
\[
\begin{split}
RSS_M = \| \bold{y} - \bold{H}_M \bold{y} \|^2 = & (\bold{X} \boldsymbol\beta^* + \boldsymbol\varepsilon)^T (\mathds{I} - \bold{H}_M)(\bold{X} \boldsymbol\beta^* + \boldsymbol\varepsilon) \\
= & \boldsymbol\beta^{*T} \bold{X}^T (\mathds{I} - \bold{H}_M) \bold{X} \boldsymbol\beta^* + 2 \boldsymbol\beta^{*T} \bold{X}^T (\mathds{I} - \bold{H}_M) \boldsymbol\varepsilon + \boldsymbol\varepsilon^T(\mathds{I} - \bold{H}_M) \boldsymbol\varepsilon.
\end{split}
\]
When $T \subseteq M$ we have $\bold{H}_M \bold{X} \boldsymbol\beta^*  = \bold{X} \boldsymbol\beta^*  $ and $RSS_M = \boldsymbol\varepsilon^T (\mathds{I} - \bold{H}_M)\boldsymbol\varepsilon$.
\begin{lemma}\label{lemmaTF}
Assuming $p \prec n$ and  $p \prec r_n$, we have
\[
\log \frac{RSS_T}{RSS_{F}} <_P \frac{r_n}{n}.
\]
\end{lemma}
\begin{proof}
Observe that
\[
\frac{RSS_T}{RSS_{F}} = 1 + \frac{RSS_T - RSS_{F}}{RSS_{F}} = 1 + \frac{p}{n}E_n,
\]
where 
\[
E_n = \frac{\boldsymbol\varepsilon^T (\bold{H}_{F} - \bold{H}_T) \boldsymbol\varepsilon}{\boldsymbol\varepsilon^T (\mathds{I} - \bold{H}_{f}) \boldsymbol\varepsilon} \cdot \frac{n}{p}.
\]
Let us notice that $\bold{H}_F - \bold{H}_T$ is 
a matrix of an orthogonal projection with rank $p - |T|$. Therefore $W_1 = \boldsymbol\varepsilon^T (\bold{H}_{F} - \bold{H}_T) \boldsymbol\varepsilon \sim \sigma^2 \chi^2_{p - |T|}$ and $W_2 = \boldsymbol\varepsilon^T (\mathds{I} - \bold{H}_{F}) \boldsymbol\varepsilon \sim \sigma^2 \chi^2_{n-p}$.
Then we get
\[
\mathds{E} \left(\frac{W_1}{p} \right) = \frac{\sigma^2(p - |T|)}{p},
\text{ Var} \left(\frac{W_1}{p} \right) = \frac{2\sigma^4(p - |T|)}{p^2}
\]
and since $p$ grows monotonically with $n$ we have either $p \xrightarrow[]{n \rightarrow \infty} \infty$, then $\text{ Var} \left(\frac{W_1}{p} \right) \xrightarrow[]{n \rightarrow \infty} 0$ and from Chebyshev's inequality $\frac{W_1}{p} \xrightarrow[]{n \rightarrow \infty} \sigma^2$ in probability or $p$ is bounded, then $\frac{W_1}{p}$ is bounded in probability. Analogously for $W_2$ we have
\[
\mathds{E} \left(\frac{W_2}{n} \right) = \frac{\sigma^2(n - p)}{n},
\text{ Var} \left(\frac{W_2}{n} \right) = \frac{2\sigma^4(n - p)}{n^2}
\]
and since $p \prec n$ from Chebyshev's inequality $\frac{W_2}{n} \xrightarrow[]{n \rightarrow \infty} \sigma^2$ in probability.

Therefore $E_n = O_P \big(1\big)$ and $\frac{RSS_T}{RSS_{F}} = 1 + O_P \bigg( \frac{p}{n} \bigg)$.
Hence
\[
\log \left( \frac{RSS_T}{RSS_{F}} \right) = \log \left( 1 + \frac{p}{n} E_n \right) \leq \frac{p}{n} E_n = O_P \bigg(\frac{p}{n}\bigg) <_P \frac{r_n}{n}.
\]
\end{proof}

\begin{lemma}\label{lemmaRSS}
Assuming that $p \prec \Delta$ ($\Delta$ is defined in equation (\ref{eq:DD})) we have for all $\delta > 1$ 
\[
\min_{M \in \mathcal{M}_\mathcal{V} } \left( \log \left( \frac{RSS_{M}}{RSS_T} \right) \right) \geq_P \log \left( 1+ \frac{\Delta}{\delta \sigma^2 \cdot n} \right).
\]
\end{lemma}
\begin{proof}
Using the fact that 
\[
\frac{1}{n} RSS_T = \frac{\boldsymbol\varepsilon^T (\mathds{I} - \bold{H}_T) \boldsymbol\varepsilon}{n} = \sigma^2 + o_P \big(1\big)
\]
and denoting
\[
RSS_{M} - RSS_T = \Delta_{M} + S_{M} + W_{T} - W_{M},
\]
where 
\[
\Delta_{M} = \boldsymbol\beta^{*T} \bold{X}^T (\mathds{I} - \bold{H}_{M}) \bold{X} \boldsymbol\beta^*, \ 
S_{M} = 2\boldsymbol\beta^{*T} \bold{X}^T (\mathds{I} - \bold{H}_{M}) \boldsymbol\varepsilon, \ 
W_{T} = \boldsymbol\varepsilon^T \bold{H}_T \boldsymbol\varepsilon \text{ and } W_{M} = \boldsymbol\varepsilon^T \bold{H}_{M} \boldsymbol\varepsilon.
\]
Note that 
\[
\Delta_{M} \geq \Delta,\ S_{M} \sim \mathcal{N} (0, 4\sigma^2 \Delta_{M}),\  W_{T} \sim \sigma^2 \chi^2_{|T|} \text{ and } W_{M} \sim \sigma^2 \chi^2_{p - 1}.
\]
Using assumption, $\frac{S_{M}}{\Delta_{M}}$, $\frac{W_{T}}{\Delta_{M}}$ and $\frac{W_{M}}{\Delta_{M}}$  are $o_P \big(1\big)$ from Chebyshev's inequality.  Since the dimension of the true model $T$ is finite and independent of $n$, so is the number of models in $\mathcal{M}_\mathcal{V}$ and we have
\[
RSS_{M} - RSS_T = \Delta_{M} \left( 1 + \frac{S_{M}}{\Delta_{M}} + \frac{W_{T}}{\Delta_{M}} - \frac{W_{M}}{\Delta_{M}} \right) = \Delta_{M} \Big(1 + o_P \big(1\big) \Big) \geq \Delta \Big(1 + o_P \big(1\big) \Big).
\]
As a result 
\[
\begin{split}
\log \frac{RSS_{M}}{RSS_T} = & \log \left( 1 + \frac{RSS_{M} - RSS_T}{RSS_T} \right) >_P \log \left( 1 + \frac{\Delta}{\delta \sigma^2 n} \right) 
\text{ for } \delta > 1. 
\end{split}
\]
\end{proof}
\begin{lemma}\label{lemma1W}
Assuming that $p \prec \Delta$  we have
\[
\max _{M \in \mathcal{M}_\mathcal{T}} \Big(\log RSS_M \Big)  <_P \min_{M \in \mathcal{M}_\mathcal{V}} \Big(\log RSS_M\Big),
\]
\end{lemma}
\begin{proof}
For $\delta > 1$ let us denote $a =  \log \left( 1 + \frac{\Delta}{\delta \sigma^2 n} \right)$, then from Lemma \ref{lemmaRSS} we get
\[
\begin{split}
\min _{M \in \mathcal{M}_\mathcal{V}} \Big(\log RSS_M \Big)  >_P & \log RSS_T + a
\geq \max _{M \in \mathcal{M}_\mathcal{T}} \Big(\log RSS_M \Big)  + a \\
\geq & \max _{M \in \mathcal{M}_\mathcal{T}} \Big(\log RSS_M \Big). 
\end{split}
\]
\end{proof}

\subsection{Ordering of squared t-statistics}

In this section we show that ordering of models $M \in \mathcal{M}_ \mathcal{T} \cup \mathcal{M}_\mathcal{V}$ with respect to squared t-statistics is equivalent to ordering them with respect to the values of residual sum of squares.

Let $t_M$, where $M \in \mathcal{M}_ \mathcal{T} \cup \mathcal{M}_\mathcal{V}$  denote t-statistic for the full model with one elementary constraint  $ \bold{A}_{0M} \beta = 0$.
\begin{lemma}\label{lT2}
If $p \prec \Delta$, then
\[
\max_{M \in \mathcal{M}_\mathcal{T}} t^2_M  <_P \min_{M \in \mathcal{M}_\mathcal{V}} t^2_M.
\] 

\end{lemma}
\begin{proof}
From Lemma \ref{lem:hhh} we get that
\[
RSS_{M} - RSS_{F} = \bold{y}^T (\bold{H}_F - \bold{H}_M) \bold{y} = \widehat{\boldsymbol\beta}^T \bold{A}_{0M}^T ( \bold{A}_{0M} (\bold{X}^T \bold{X})^{-1}  \bold{A}_{0M}^T)^{-1}  \bold{A}_{0M} \widehat{\boldsymbol\beta},
\]
where $\widehat{\boldsymbol\beta} = (\bold{X}^T\bold{X})^{-1}\bold{X}^T\bold{y}$.
Hence for for each $M \in \mathcal{M}_ \mathcal{T} \cup \mathcal{M}_\mathcal{V}$ 
\[
\begin{split}
t^2_M = & \frac{( \bold{A}_{0M} \widehat{\boldsymbol\beta})^2}{\widehat{\text{Var}}( \bold{A}_{0M} \widehat{\boldsymbol\beta})} = \frac{( \bold{A}_{0M} \widehat{\boldsymbol\beta})^2}{ \bold{A}_{0M} \widehat{\text{Var}}(\widehat{\boldsymbol\beta})  \bold{A}_{0M}^T} = \frac{( \bold{A}_{0M} \widehat{\boldsymbol\beta})^2}{\widehat{\sigma}^2  \bold{A}_{0M} (\bold{X}^T \bold{X})^{-1}  \bold{A}_{0M}^T} = \frac{RSS_{M} - RSS_{F}}{\widehat{\sigma}^2},
\end{split}
\]
where $\widehat{\sigma}^2 = \frac{RSS_{F}}{n - |F|}$. Observe that $ \bold{A}_{0M}$ is $1 \times |F|$ matrix, thus
\[
t_M^2 = (n - |F|)\frac{RSS_{M} - RSS_{F}}{RSS_{F}},
\]
and from Lemma \ref{lemma1W} we get the conclusion.
\end{proof}

\subsection{Correct ordering of constraints using hierarchical clustering}
In this subsection we state conditions under which the true model $T$ belongs to the path of nested models obtained in step 4 of DMR algorithm. 

Temporarily let us limit the analysis to a model consisting of one factor and no continuous variables. The true partition of set $\{ 1,\ldots, p_1 \}$ will be denoted by $P^*_1 = ( C_{i1}^* )_{i=1}^{|T|}$. We say that distance matrix $\bold{D} = [d_{ij}]_{ij}$ is consistent with the true partition if dissimilarity measures for elements within the same clusters are smaller than for elements from different clusters:
\begin{equation}\label{eqDISS}
\max_{l \in \{1, \ldots, |T|\}} \max_{i,j \in C^*_{l1}} d_{ij} = d^{true} < d^{false} =  \min_{\substack{l_1, l_2 \in \{1, \ldots, |T|\} \\ l_1 \neq l_2}} \min_{i \in C^*_{l_11}, j \in C^*_{l_21}} d_{ij}.
\end{equation}
Let $P_{s1} = ( C_{is1} )_{i=1}^{p_1 - s +1}$ denote a partition of set $\{1, \ldots, p_1\}$ in step  $s$ of hierarchical clustering algorithm, $s=1, \ldots, p_1$. We will name aggregation of $C_{i_ss1}$ and $C_{j_ss1}$  in step $s$ compatible with the true partition $P^*_1$ if there exist $l \in \{ 1, \ldots, |T| \}$, $i_{s+1} \in \{ 1,\ldots, p_1 - s \}$ and $i_s \neq j_s$, $i_s, j_s \in \{ 1, \ldots p_1 - s+1 \}$ such that
\[
C_{i_{s+1}s+11}= C_{i_ss1} \cup C_{j_ss1} \ ,\ C_{i_{s+1}s+11} \subseteq C^*_{l1}.
\]
Cutting height in step $s$ is defined as $h_{s1} = d(C_{i_ss1}, C_{j_ss1})$ if $C_{i_ss1}$ and $C_{j_ss1}$ are aggregated in this step, $\bold{h}_1 = \left( h_{11}, \ldots, h_{p_1-1,1} \right)$.
\begin{lemma}\label{lclust}
Assuming that the linkage criterion of hierarchical clustering algorithm satisfies:
\begin{equation}\label{eqMONO}
\begin{split}
d &\left( C_{i_{s+1}s+1k} = C_{i_ssk}  \cup C_{j_ssk} , C_{j_{s+1}s+1k} = C_{o_ssk}  \right) \\
&= b \min \left\{ d \left( C_{i_ssk}, C_{o_ssk}  \right), d \left( C_{j_ssk}, C_{o_ssk}  \right)  \right\} \\
&+ (1-b)\max \left\{ d \left( C_{i_ssk}, C_{o_ssk}  \right), d \left( C_{j_ssk}, C_{o_ssk}  \right)  \right\},
\end{split}
\end{equation}
where $b \in [0,1] $
and the dissimilarity matrix has property (\ref{eqDISS}), then the cutting heights for aggregations compatible with $P^*_1$ are lower than $d^{true}$ and cutting heights for aggregations not compatible with $P^*_1$ are larger than $d^{false}$.
\end{lemma}
\begin{proof}
From  (\ref{eqDISS}) if $|T| = p_1$  the statement holds trivially and if $|T| < p_1$ aggregation in the first step is compatible with $P^*_1$. We assume that in step $s$ aggregation is compatible with the true partition with cutting height not greater than $d^{true}$. 
If aggregation of $C_{i_{s+1}s+1,1}= C_{i_ss1} \cup C_{j_ss1} $ and $C_{j_{s+1}s+1,1}= C_{o_ss1}  $ is compatible with $P^*_1$ then
\[
h_{s1} = d \left( C_{i_{s+1}s+11}, C_{j_{s+1}s+11} \right) \leq \max \left( d \left( C_{i_ss1}, C_{o_ss1} \right) , d \left( C_{j_ss1} , C_{o_ss1}\right)  \right) \leq d^{true}
\]

If aggregation of $C_{i_{s+1}s+11}= C_{i_ss1} \cup C_{j_ss1} $ and $C_{j_{s+1}s+11}= C_{o_ss1}  $ is not compatible with $P^*_1$ then
\[
h_{s1} = d \left( C_{i_{s+1}s+11}, C_{j_{s+1}s+11} \right) \geq \min \left( d \left( C_{i_ss1}, C_{o_ss1} \right) , d \left( C_{j_ss1} , C_{o_ss1}\right)  \right)  \geq d^{false}
\]

Hence,  cutting heights $h_{11}, \ldots, h_{p_1-|T|,1}$ not greater than $d^{true}$ are used until all aggregations compatible with $P^*_1$ are performed.  We have $C_{ p_1 - |T| + 11} = P^*_1$ and in steps $s = p_1 - |T| + 2, \ldots, p_1$ the true partition $P^*_1$ is a subpartition of $C_{s1}$ and cutting heights  $h_{p_1-|T|+11}, \ldots, h_{p_1 - 11} $  are not less than $d^{false}$. 
\end{proof}
Note that linkage criteria: single, complete and average satisfy assumption (\ref{eqMONO}).

\begin{proof}[\textbf{\emph{Proof of Theorem 1a}}] 

Let us denote the path of nested models from step 4 of DMR algorithm by $J = \{M_0, \ldots, M_{p - 1} \}$.
The event of erroneous selection of the model by DMR algorithm is a subset of a sum of three events:
\[
\begin{split}
\{ \widehat{T} \neq T \} &\subseteq \{ T \notin  J \} \cup \{ T \in  J,  \text{GIC}_T \geq \min_{M \subsetneq T  } \text{GIC}_M \}\\
 & \cup \{T \in  J, \text{GIC}_T \geq \min_{T \subsetneq M  } \text{GIC}_M \} \\
&\subseteq \{ T \notin J\} \cup \{ \text{GIC}_T \geq \min_{M \subsetneq T  } \text{GIC}_M \}
 \cup \{ \text{GIC}_T \geq \min_{T \subsetneq M  } \text{GIC}_M \}.
\end{split} 
\]
We will show that the probability of each of them tends to zero when $n \rightarrow \infty$.

Using Lemma \ref{lT2} let us consider constant $h_*$ such that
\[
\max_{M \in \mathcal{M}_\mathcal{T}} t^2_M <_P h_*  <_P \min_{M \in \mathcal{M}_\mathcal{V}} t^2_M.
\] 
It is obvious that cutting heights for true constraints for continuous variables are smaller than $h_*$ and for false ones greater than $h_*$. It also follows from Lemma \ref{lT2} that dissimilarity matrices used in the algorithm are consistent with the partitions for model $T$. Then, applying  Lemma \ref{lclust} for each factor, we get that the cutting heights for aggregations compatible with the true partitions are not greater than $h_*$ and for incompatible ones not smaller than $h_*$. Hence, in DMR algorithm accepting true constraints precede accepting false ones, for large $n$ the probability that the true model lies on the path of nested models tends to 1. 

Since  $\min_{T \subsetneq M } RSS_M \geq RSS_F$ we have 
\[
 \{ \text{GIC}_T \geq \min_{T \subsetneq M  } \text{GIC}_M \} \subseteq  \{ \log RSS_T \geq \log RSS_F + \frac{r_n}{n} \}
\]
and  from Lemma \ref{lemmaTF} we know that
$$
\mathbb{P} \left(  \log RSS_T \geq \log RSS_F + \frac{r_n}{n}  \right)   \xrightarrow[]{\ P \ } 0.
$$
It is obvious that
\[
 \{ \text{GIC}_T \geq \min_{M \subsetneq T  } \text{GIC}_M \} \subseteq  \{  \log RSS_T \geq  \min_{M \in \mathcal{M}_\mathcal{V} } \log RSS_M - \frac{|T|r_n}{n} \}.
\]
Let us notice from assumptions of theorem that $\frac{|T|r_n}{n} \prec \frac{\Delta}{\delta \sigma^2 n + \Delta} \leq \log \left( 1+\frac{\Delta}{\delta \sigma^2 n } \right)$. Then
\[ 
\left\{  \frac{|T|r_n}{n} \geq  \min_{M \in \mathcal{M}_\mathcal{V} } \log  \frac{RSS_M}{RSS_T} \right\} \supseteq \left\{  \log \left( 1+\frac{\Delta}{\delta \sigma^2 n } \right) \geq  \min_{M \in \mathcal{M}_\mathcal{V} } \log \frac{RSS_M}{RSS_T}  \right\}
\]
and  from Lemma \ref{lemmaRSS} we know that
\[
\mathbb{P} \left(  \log \left( 1+\frac{\Delta}{\delta \sigma^2 n } \right) \geq  \min_{M \in \mathcal{M}_\mathcal{V} } \log  \frac{RSS_M}{RSS_T}    \right)   \xrightarrow[]{\ P \ } 0.
\]

Hence, DMR algorithm is a consistent model selection method.
\end{proof}
\begin{proof}[\textbf{\emph{Proof of Theorem 1b}}]
Let us denote
\[
\bold{g}_n = \sqrt{n}\left( \widehat{\boldsymbol\beta}_{T} - \boldsymbol\beta^* \right) \text{ and } \bold{b}_n =  \sqrt{n}\left( \widehat{\boldsymbol\beta}_{\widehat{T}} -  \boldsymbol\beta^*  \right),
\] 
Notice that $\bold{g}_n = \bold{b}_n$ 
if $\widehat{T} = T$. From Theorem 1a
\[
 \mathbb{P} \left(\mathds{1}(\widehat{T} \neq T) = 0 \right) \xrightarrow[]{\ P \ } 1.
\]
Since
\[
\left\{ \mathds{1}(\widehat{T} \neq T) = 0 \right\} \subseteq  \left\{ \bold{b}_n \mathds{1}(\widehat{T} \neq T) = 0 \right\},
\]
hence $\bold{b}_n \mathds{1}(\widehat{T} \neq T) \xrightarrow[]{\ P \ } 0$. From properties of the OLS estimator we have 
\[
\bold{g}_n \mathds{1}(\widehat{T} = T) \xrightarrow[]{\ d \ } \mathds{N}(0, \sigma^2 \bold{\Sigma}_T).
\]
Henceforth, from multidimensional Slutsky's theorem we get
\[
\bold{b}_n = \bold{b}_n \mathds{1}(\widehat{T} \neq T) + \bold{b}_n \mathds{1}(\widehat{T} = T) = \bold{b}_n \mathds{1}(\widehat{T} \neq T) + \bold{g}_n \mathds{1}(\widehat{T} = T) \xrightarrow[]{\ d \ } \mathds{N}(0, \sigma^2 \bold{\Sigma}_T).
\]
\end{proof}


\bibliographystyle{imsart-nameyear}

\bibliography{DMR}

\end{document}